\documentclass[12pt]{article}
\pdfoutput=1
\usepackage[T1]{fontenc}
\usepackage{lmodern}
\usepackage{xurl} % allow line breaks in long URLs / \n\Urlmuskip=0mu plus 1mu
\usepackage[nopatch=footnote]{microtype}
\usepackage{setspace}
\usepackage[colorlinks=true, pdfstartview=Fit, linkcolor=blue,citecolor=blue, urlcolor=blue]{hyperref}
\usepackage[margin=1in]{geometry}
\usepackage[title]{appendix}
\usepackage[round]{natbib}
\usepackage[dvipsnames]{xcolor}
\usepackage{amsmath,amsthm,amssymb,ifthen,mathrsfs,amsfonts,bm,bbm}
\usepackage{cases}
\usepackage{multirow}
\usepackage{empheq}
\usepackage{mathtools}
\usepackage{mathabx}
\usepackage{subcaption}
\usepackage{graphicx}
\usepackage{accents}
\usepackage{enumitem}
\usepackage{bigints}
\usepackage{bookmark}
\usepackage{caption}
\usepackage[capposition=below]{floatrow}

\usepackage[framemethod=tikz]{mdframed}
\usepackage{pgf,tikz}
\usetikzlibrary{calc,math}
\usetikzlibrary{arrows,shapes,positioning,automata}
\usetikzlibrary{decorations.markings}
\usetikzlibrary{shapes.geometric}
\usepackage{pgfplots}
\usetikzlibrary{intersections}
\pgfplotsset{compat=newest}
\usepgfplotslibrary{fillbetween}
\usetikzlibrary{patterns}
\newtheorem{theorem}{Theorem}
\newtheorem{lemma}{Lemma}
\newtheorem{proposition}{Proposition}
\theoremstyle{definition}
\newtheorem{corollary}{Corollary}
\newtheorem{definition}{Definition}

\newtheorem{remark}{Remark}

\newcommand{\BE}{\mathsf{E}}
\newcommand{\BP}{\mathsf{P}}
\newcommand{\BR}{\mathbb{R}}
\newcommand{\BN}{\mathbb{N}}

\renewcommand{\d}{\mathrm{d}}
\def\ee{\mathrm{e}}

\newcommand{\ve}{\varepsilon}

\makeatletter
\newcommand{\rn}[1]{\romannumeral #1}
\newcommand{\Rn}[1]{\expandafter\@slowromancap\romannumeral #1@}
\makeatother

\numberwithin{equation}{section}
\title{The Impact of Connectivity on the Production and Diffusion of Knowledge}
\author{}
\date{November 2026}

\author{Gustavo Manso\,$^{\dagger}$ and Farzad Pourbabaee\,$^{\ddagger}$
\thanks{$^{\dagger}$Haas School of Business, University of California Berkeley: \href{mailto:manso@berkeley.edu}{\nolinkurl{manso@berkeley.edu}}.}
\thanks{$^{\ddagger}$Barry S. Miller College of Business, University of Central Florida: \href{mailto:farzad.pourbabaee@ucf.edu}{\nolinkurl{farzad.pourbabaee@ucf.edu}}.}}

\begin{document}
\let\savedfootnotemark\footnotemark
\renewcommand\footnotemark{}
\maketitle
\let\footnotemark\savedfootnotemark

\begin{abstract}
    We study a social bandit problem featuring the production and diffusion of knowledge. While higher connectivity enhances knowledge diffusion, it may reduce knowledge production as agents shy away from experimentation and free-ride on others' discoveries. In the symmetric mixed-strategy equilibrium, exploration probability falls with connectivity, while expected per-capita equilibrium surplus is constant in the interior mixing region. Under finite-player pure strategies, surplus can fall discontinuously when the integer number of explorers drops. In the large-population limit these jumps vanish, and per-capita equilibrium surplus is weakly increasing in connectivity and eventually flat for intermediate beliefs.
\end{abstract}
\clearpage

\setstretch{1.35}
\tableofcontents
\setcounter{tocdepth}{3}
\clearpage

\interfootnotelinepenalty=10000
\medmuskip= 0.5mu plus 1.0mu minus 1.0mu

\section{Introduction}

Advances in travel and communication technologies have cleared the way for more connected organizations and societies. In well-connected structures, new ideas spread quickly leading to rapid innovation adoption. 

While such enhanced knowledge diffusion is beneficial, it may come at the cost of reduced knowledge production. When an organization or society is well-connected, agents may shy away from experimentation with new ideas, since they can easily see the results of the experimentation efforts of other agents and adapt their actions accordingly. Because of this free-riding, more connected organizations or societies may produce less new knowledge and become more homogeneous, even while they diffuse existing knowledge more effectively.

We study this tension between knowledge diffusion and knowledge production in a simple two-period social bandit model. In each period, each agent has the choice between exploiting a safe well-known action or exploring a risky novel action. The payoff from the risky action depends on an unknown payoff-relevant state $\theta$, which may be favorable or unfavorable and determines whether the alternative can generate a high return. Agents share an initial belief $\pi$, the probability they assign to $\theta$ being the favorable state, and update it after observing experimentation outcomes. At the end of the first period, each agent observes the outcome of a randomly selected group of agents. In the general model of connections, a first-order stochastic increase in contacts raises the full-exploration cutoff, while a mean-preserving spread in the distribution of contacts lowers the cutoff and encourages more exploration.

The leading model of connections in the paper is the Erd\H{o}s-R\'enyi economy, in which each pair of agents is connected with the same probability, and each agent observes the outcomes of her immediate neighbors. We show first that connectivity has a clear negative effect on knowledge production. In this environment, holding the prior belief fixed in the interior symmetric mixed-strategy region, the exploration probability decreases with connectivity, while expected per-capita equilibrium surplus is constant. The welfare consequences under pure strategies require more care: at finite $n$, welfare can fall when the integer number of explorers drops. These discontinuous changes vanish in the large-$n$ Erd\H{o}s-R\'enyi limit, where per-capita surplus is weakly increasing in connectivity and, for intermediate beliefs, becomes flat once connectivity is sufficiently high.

\paragraph{Main results.} We organize the paper around the effect of connectivity at a fixed prior belief $\pi$. ``More connectivity'' means a larger chance of observing another agent in the two-player model, a first-order stochastic increase in the contact variable $M$ in the general model, and a larger average degree $\lambda$ in the Erd\H{o}s-R\'enyi economies. The main conclusions are as follows:
\begin{enumerate}[label=(\roman*)]
    \item \emph{Connectivity lowers experimentation.} A first-order stochastic increase in contacts raises the full-exploration threshold, while a mean-preserving spread in the distribution of contacts lowers the threshold and encourages more exploration (Theorem~\ref{thm: exploration_equil}). In the local Erd\H{o}s-R\'enyi economy, holding an interior prior fixed, the finite-$n$ symmetric mixed-strategy exploration probability decreases with $\lambda$ (Proposition~\ref{prop:equil_mix_prob}).
    \item \emph{The symmetric mixed equilibrium and the large-population economy provide the welfare benchmark.} In finite economies, the welfare benchmark is the symmetric mixed equilibrium: expected per-capita equilibrium surplus in the interior mixing region is constant in connectivity (Corollary~\ref{cor:mixed_expected_utility}). In the large-$n$ Erd\H{o}s-R\'enyi limit, per-capita equilibrium surplus under pure strategies and expected per-capita surplus under symmetric mixed strategies converge to the same limit (Proposition~\ref{prop: asympt_equil_soc_surplus} and Theorem~\ref{thm:ubar_mixed_infty}). This common limit is weakly increasing in connectivity and eventually flat for intermediate beliefs.
    \item \emph{Finite-player pure-strategy welfare is a qualification.} Surplus can drop when a small increase in connectivity induces one fewer agent to explore (Proposition~\ref{prop: equil_soc_surplus_discont}). This is a discrete-population effect that diminishes in the large-population limit.
    \item \emph{Connectivity creates an efficiency wedge.} Agents do not internalize the information their experiments provide to others. Equilibrium therefore features over-exploitation and under-exploration relative to the planner (Theorem~\ref{thm: social_optimum}), and under symmetric mixing the equilibrium probability is smaller than the optimal level, i.e., $\mu^e\leq \mu^*$ (Corollary~\ref{cor:equil_vs_optimal_mu}).
\end{enumerate}
We therefore treat the symmetric mixed-strategy equilibrium and the large-$n$ local Erd\H{o}s-R\'enyi economy as the clean welfare benchmark. The asymmetric pure-strategy equilibrium remains useful because it identifies the equilibrium number of explorers and explains the finite-population welfare drops. The planner's problem is used to measure the under-exploration wedge created by free-riding.

The paper proceeds as follows. Section~\ref{sec: 2player_Econ} develops the mechanism in the two-player economy. Section~\ref{sec: many_players_Econ} characterizes pure and symmetric mixed equilibria in the $n$-agent model under a general contact distribution and records the first-order and second-order stochastic comparative statics. Section~\ref{sec: limit} specializes the equilibrium analysis to \emph{local} and \emph{global} Erd\H{o}s-R\'enyi economies, where the average degree $\lambda$ is the connectivity parameter. Section~\ref{sec: social_surplus} studies welfare in the local Erd\H{o}s-R\'enyi economy, separating finite-player pure-strategy drops from the large-population and mixed-equilibrium welfare benchmark, and then compares equilibrium experimentation with the planner's choice. Section~\ref{sec:conclusion} concludes.

%------------------------------------------
\paragraph{Related Literature.} In his seminal work, ~\cite{rothschild1974two} studies the \textit{single-agent} experimentation problem in the two-armed bandit environment and shows that the agent settles on the sub-optimal arm with positive probability. The literature on multi-agent strategic experimentation begins with \cite{bolton1999strategic} and \cite{keller2005strategic},\footnote{A non-exhaustive list of related papers in strategic bandits includes \cite{heidhues2015strategic}, \cite{keller2015breakdowns}, \cite{bonatti2017learning}, and \cite{pourbabaee2020robust}.} where players are completely connected and each observes the outcomes of \textit{all} others. Our paper interpolates between these two extremes by considering agents who are neither fully connected nor fully isolated, allowing us to explicitly study how connectivity shapes exploration incentives.

While \cite{keller2005strategic} show that the exploration threshold increases with the number of players under complete observability, we show that it increases with connectivity under partial observability. Varying connectivity while holding the population fixed yields comparative statics unavailable when complete observability mechanically links the two. At finite $n$, pure-strategy regime changes generate discrete welfare drops; in the interior symmetric mixed equilibrium, exploration falls with connectivity while expected per-capita equilibrium social surplus remains constant (Corollary~\ref{cor:mixed_expected_utility}); and in the large-population local Erd\H{o}s-R\'enyi benchmark, these drops vanish and per-capita welfare is weakly increasing and eventually flat for intermediate beliefs. We also study mean-preserving spreads (second-order stochastic dominance) in the distribution of contacts, a comparative static unavailable under deterministic observability.

Our work is closely related to two papers studying experimentation in networks. \cite{board2024experimentation} show that social surplus is single-peaked in network density. We also study how network density changes experimentation and welfare, but our Erd\H{o}s-R\'enyi specification ($p=\lambda/n$) separates the finite-population pure-strategy effect from the large-population benchmark. In finite populations, surplus drops can occur at regime changes when a discrete explorer is lost. \cite{salish2015learning} studies a multi-period version of our model with richer strategic dynamics. We opt for a two-period framework, which allows us to derive closed-form exploration thresholds, characterize the limiting fraction of explorers as $n\to\infty$, and provide complete equilibrium-versus-optimum welfare comparisons—while still capturing the core tensions between knowledge production and diffusion.

\cite{bala1998learning}, \cite{gale2003bayesian} and \cite{sadler2020innovation} study the social learning dynamics of \textit{myopic} agents who are connected in networks and collect information from their neighbors to maximize their \textit{short-run} payoff. Our two-period experimentation framework is a first attempt to depart from these works by allowing agents to have long-run incentives in their strategic interactions.

Topics such as long-run social conformity and information aggregation in multi-agent strategic experimentation—when agents observe actions rather than payoffs—are studied in \cite{chamley1994information}, \cite{aoyagi1998mutual}, \cite{rosenberg2007social}, \cite{rosenberg2009informational}, and \cite{camargo2014learning}. Beyond the difference in observability, our paper focuses on how connectivity shapes equilibrium strategies and social welfare, rather than on long-run conformity or social learning.

Our paper also relates to the broader literature on games with information sharing and externalities. \cite{duffie2009information} studies a continuum economy where individuals endowed with informative signals incur costly search to meet and share information. \cite{wolitzky2018learning} investigates social learning and innovation adoption where agents arrive continuously over time, learn from a random sample of past outcomes, and make once-and-for-all decisions. In a Poisson news setting, \cite{frick2024innovation} studies how the arrival rate of a public signal—which depends on the mass of current adopters—affects innovation adoption. Somewhat related to the tension between private experimentation and information sharing, \cite{gordon2021observation} show that increasing the observation delay among team members working on a joint project raises individual effort. In collective search settings, \cite{lomys2023collective} studies the relationship between learning outcomes and the structure of the connection graph, and~\cite{cetemen2023collective} study how the outcomes depend on the members' complementarity and alliances. In moral hazard settings, \cite{bonatti2011collaborating} show that collaboration among agents diminishes over time and that the expected time to success increases with the number of agents. Lastly, the analysis of our paper on how connectivity impacts exploration incentives has implications for designing optimal policies to motivate innovation and exploration in networked economies \citep[see, ][]{manso2011motivating, kerr2014entrepreneurship}.
%%%%%%%%%%%%%%%%%%%%%%%%%%%%
%%%%%%%%%%%%%%%%%%%%%%%%%%%%%%%%%%%%%%%%%%%%%%%%%%%%%%%%%%%%%%%%%%%%%%%%%%%%%%%%%%%%%%%%%%%%%%%%%%%%%%%%%%%%%%
\section{Two-Player Economy}
\label{sec: 2player_Econ}
In this section, we use a simple two-agent model to illustrate the equilibrium forces that arise with $n>2$ agents. Greater communication improves second-period learning but weakens first-period incentives to experiment. As a result, agents under-explore relative to the planner, and finite-player welfare can be non-monotonic when the equilibrium number of explorers changes discretely. Propositions~1-3 follow from the general case in Sections~\ref{sec: many_players_Econ} and~\ref{sec: social_surplus}, so we omit their proofs here.
\subsection{Perfect Connections}
There are two agents $i$ and $j$, and the game consists of two periods, $t \in \{0,1\}$. Each agent faces a binary action choice in each period: a safe action ($a=0$, exploiting the status quo) with normalized payoff $0$, or a risky action ($a=1$, exploring the alternative). For the risky action, the return is a binary random variable $y \in \{-\alpha,1\}$ (with $\alpha \in (0,1)$) conditioned on the hidden state $\theta \in \{0,1\}$, with the following structure:
\begin{equation*}
    \BP\left(y=1 \mid \theta=1\right) = \beta \in (0,1), \text{ and } \BP\left(y=1 \mid \theta = 0\right) = 0\,.
\end{equation*}
Therefore, receiving a high payoff $y=1$ perfectly reveals the underlying state $\theta$, as the probability of a high payoff in the low state is zero. The realized payoff also reveals the action choice: observing a payoff of $0$ (respectively, nonzero) implies the agent took the safe (respectively, risky) action.

Let $\pi = \BP(\theta=1)$ be the initial prior of both players. The following timeline elucidates the order of events in this two-period economy:
\begin{figure}[ht]
\begin{center}
\begin{tikzpicture}[decoration=snake]
\node (start) at (0,0) {};
\node (end) at (15,0) {};
\path [thick,->,-stealth] (start) edge (end);

\node (s0) at (1,-0.25) {};
\node (t0) at (1,0.25) {};
\path [-] (s0) edge (t0);
\node [anchor=south] at (t0.north) {$t=0$};

\node (s1) at (6,-0.25) {};
\node (t1) at (6,0.25) {};
\path [-] (s1) edge (t1);
\node [anchor=south] at (t1.north) {$t=1$};

\draw [semithick, ->, decorate] (2,0) -- (2,-0.7);
\node[align=center, font=\scriptsize]  (a0) at (2,-1.3) {action\\$a_0 \in \{0,1\}$};

\draw [semithick, ->, decorate] (4.5,0) -- (4.5,-0.7);
\node[align=center, font=\scriptsize]  (y0) at (4.5,-1.3) {payoff\\realization $y_0$\\if $a_0=1$};

\draw [semithick, ->, decorate] (7,0) -- (7,-0.7);
\node[align=center, font=\scriptsize]  () at (7,-1.3) {observing\\ the payoff\\of the other player};

\draw [semithick, ->, decorate] (9.5,0) -- (9.5,-0.7);
\node[align=center, font=\scriptsize]  () at (9.5,-1.3) {updating\\prior $\pi$\\to posterior $\tilde{\pi}$};

\draw [semithick, ->, decorate] (11.5,0) -- (11.5,-0.7);
\node[align=center, font=\scriptsize]  () at (11.5,-1.3) {action\\$a_1\in \{0,1\}$};

\draw [semithick, ->, decorate] (13.5,0) -- (13.5,-0.7);
\node[align=center, font=\scriptsize]  () at (13.5,-1.3) {final payoff\\ realization $y_1$\\if $a_1=1$};
\end{tikzpicture}
\caption{Timeline of the two-period bandit}
\label{fig: timeline}
\end{center}
\end{figure}

At the beginning of the second period, agent $i$ observes the outcome of agent $j$'s first period experimentation if $j$ chose the risky arm. This communication step is the main focus throughout the paper. Agent $i$ then updates her prior about $\theta$ given $y_0(i)$ and $y_0(j)$, yielding posterior $\tilde{\pi}$. Let $\pi_{\ell,m}$ denote the posterior when agent $i$ observes $m \in \{0,1,2\}$ payoff signals, of which $\ell \in \{0,1,2\}$ are high realizations ($y=1$), where $\ell \leq m$. Then,
\begin{equation*}
    \pi_{\ell,m} = 1_{\{\ell\geq 1\}}+\frac{\pi(1-\beta)^m}{1-\pi+\pi(1-\beta)^m} \,1_{\{\ell=0\}}\,.
\end{equation*}
Before observing $y_0(i)$ and $y_0(j)$, $\tilde \pi$ is a random variable with ex post values $\pi_{\ell, m}$.

The game ends with each agent making a second action choice. Since the safe option is always available, the expected payoff \textit{after} Bayesian updating with posterior $\tilde \pi$ is
\begin{equation*}
    \big(\tilde{\pi}\beta-\alpha(1-\tilde{\pi}\beta )\big)^+ \coloneq \max\{\tilde{\pi}\beta(1+\alpha)-\alpha,0\}\,.
\end{equation*}
Thus, unlike the strategic first period choice $a_0$, the second period decision $a_1$ follows a simple threshold rule: choose the risky arm if and only if $\tilde \pi > \frac{\alpha}{\beta(1 + \alpha)}$ (ties favor the safe arm). 
\begin{definition}
Denote the second period's exploration threshold by $\tau \coloneq \frac{\alpha}{\beta(1+\alpha)}$. We maintain $\tau\in(0,1)$, or equivalently $\beta-\alpha(1-\beta)>0$.
\end{definition}

There will be two types of \textit{symmetric} pure-strategy equilibria: \textit{exploration} equilibrium in which both agents choose the risky arm in the first period (i.e., $a_0(i) = a_0(j) = 1$), and \textit{exploitation} equilibrium where both agents select the safe arm in the first period (i.e., $a_0(i) = a_0(j) = 0$). The equilibrium is called \textit{asymmetric} when one agent explores and the other one exploits (thus $a_0(i) \neq a_0(j)$). Let $\delta \in (0,1)$ be the time discount factor, that is each agent values the payoffs in the first and second periods with the respective weights of $1-\delta$ and $\delta$. This means that our agents are not myopic and they incorporate future gains from current exploration in their decision problem.\footnote{This is in contrast to the social learning models of \cite{bala1998learning} and \cite{sadler2020innovation} in which players are myopic. In particular, they collect information from their neighbors just to maximize their \textit{present} period payoff.}
\begin{proposition}
\label{prop: 2_players_full_sharing}
There exist two thresholds $\underline{\pi} < \bar{\pi}$ such that the exploitation equilibrium appears on $[0,\underline{\pi}]$, and the exploration equilibrium appears on $(\bar{\pi},1]$. In the intermediate region $(\underline{\pi},\bar{\pi}]$ the equilibrium is asymmetric with only one agent exploring. Closed form expressions for the cutoffs are
\begin{equation}
\label{eq: full_sharing_thresholds}
    \underline{\pi}=\frac{\alpha(1-\delta)}{\beta\left((1+\alpha)(1-\delta)+\delta (\beta -\alpha(1-\beta))\right)}\,, \quad \bar{\pi}=\frac{\alpha(1-\delta)}{\beta\left((1+\alpha)(1-\delta)+\delta (\beta -\alpha (1-\beta))(1-\beta) \right)}\,.
\end{equation}
\end{proposition}
This result shows that the equilibrium number of explorers is weakly increasing in the initial belief. Two important comparative statics are the effects of patience ($\delta$) and signal precision ($\beta$) on these thresholds.
\begin{figure}[htp]
\begin{center}
\begin{subfigure}{.45\textwidth}
\begin{center}
\pgfplotsset{normalsize}
\begin{tikzpicture}
\tikzmath{
    \A = 0.5;
    \Be = 0.8;
}
\begin{axis}[smooth,
    axis line style={semithick,-stealth},
    xtick={0,1},
    xticklabels={$0$,$1$},
    x label style={at={(axis description cs:1,0)},anchor=south},
    ytick={0},
    yticklabels={},
    xmax=1.1,
    ymax=0.45,
    axis lines=left,
    xlabel={$\delta$},
    x label style={below},
    ]
\addplot[forget plot,domain=0:1,samples = 200,dashed,thick]
{(\A*(1-x))/(\Be*((1+\A)*(1-x)+(x*(\Be-\A*(1-\Be)))))}[xshift=5pt] node [pos=0.5,above] {$\underline{\pi}$};

\addplot[forget plot,domain=0:1,samples = 400,thick]
{(\A*(1-x))/(\Be*((1+\A)*(1-x)+(x*(\Be-\A*(1-\Be))*(1-\Be))))}[xshift=5pt] node [pos=0.5,above] {$\bar{\pi}$};
\end{axis}    
\end{tikzpicture}
\caption{Effect of $\delta$}
\label{fig: patience}
\end{center}
\end{subfigure}\quad\quad
\begin{subfigure}{.45\textwidth}
\begin{center}
\pgfplotsset{normalsize}
\begin{tikzpicture}
\tikzmath{
    \A = 0.5;
    \D = 0.9;
}
\begin{axis}[smooth,
    axis line style={semithick,-stealth},
    xtick={0.333,1},
    xticklabels={$\frac{\alpha}{1+\alpha}$,$1$},
    x label style={at={(axis description cs:1,0)},anchor=south},
    ytick={0},
    yticklabels={},
    xmax=1.1,
    ymax=1.1,
    axis lines=left,
    xlabel={$\beta$},
    x label style={below},
    ]
\addplot[forget plot,domain=0.333:1,samples = 400,dashed,thick]
{(\A*(1-\D))/(x*((1+\A)*(1-\D)+(\D*(x -\A*(1-x)))))}[xshift=-5pt] node [pos=0.5,below] {$\underline{\pi}$};

\addplot[forget plot,domain=0.333:1,samples = 200,thick]
{(\A*(1-\D))/(x*((1+\A)*(1-\D)+(\D*(x -\A*(1-x)))*(1-x)))}[xshift=5pt] node [pos=0.5,above] {$\bar{\pi}$};
\end{axis}    
\end{tikzpicture}
\caption{Effect of $\beta$}
\label{fig: precision}
\end{center}
\end{subfigure}
\caption{Comparative statics of thresholds}
\label{fig: comparative_static_thresholds}
\end{center}
\end{figure}
As Figure~\ref{fig: patience} shows, higher patience (namely higher $\delta$) is associated with smaller exploration thresholds, thereby increasing the incentives to sacrifice the present payoff to learn about the risky arm and gain the benefits in the next period. Specifically, higher patience enlarges the exploration equilibrium region and shrinks the exploitation region. 
%-----------------------------------------------------------

Higher uncertainty about the risky arm (namely intermediate levels of $\beta$) is associated with higher gains from exploration, and hence a lower exploration threshold. Figure~\ref{fig: precision} confirms this intuition. In addition, higher $\beta$ increases the exploration gain upon receiving conclusive signals about $\theta$ more so than it raises the opportunity cost of exploration absent such signals. Therefore, it lowers the individual's incentive to exploit the safe arm, thereby shrinking the exploitation region (see $\underline{\pi}$ in Figure~\ref{fig: precision}).
\subsection{Imperfect Connections}
Suppose the connection between the two players is \textit{imperfect}. That is, each agent observes the outcome of the other agent's first period experimentation, and thereby their action choice, with probability $p$. In the next proposition, we show that such imperfect communication does not impact the exploitation region and \textit{expands} the exploration region.
\begin{proposition}
\label{prop: 2_players_imperfect_comm}
In the presence of imperfect connections ($p<1$), there exist two thresholds $\underline{\pi}<\bar{\pi}$ such that the exploitation equilibrium appears on $[0,\underline{\pi}]$, and the exploration equilibrium appears on $(\bar{\pi},1]$. In the intermediate region $(\underline{\pi},\bar{\pi}]$ the equilibrium is asymmetric with only one agent exploring. Closed form expressions for the cutoffs are
\begin{equation*}
    \underline{\pi}=\frac{\alpha(1-\delta)}{\beta\left((1+\alpha)(1-\delta)+\delta (\beta -\alpha(1-\beta))\right)}\,, \quad \bar{\pi}=\frac{\alpha(1-\delta)}{\beta\left((1+\alpha)(1-\delta)+\delta (\beta -\alpha (1-\beta))(1-p\beta) \right)}\,.
\end{equation*}
\end{proposition}
The key result is that $\d \bar{\pi}/\d p>0$: weaker ties correspond to higher exploration in this two-player economy. Stronger connections increase free-riding incentives, reducing first period exploration and thus raising the belief threshold required for exploring the risky arm.
%----------------------------------------------------------
\subsection{Social Surplus}
So far we have analyzed the equilibrium response in the two-player bandit game with imperfect connections. We now investigate when a ``benevolent'' planner prescribes exploitation or exploration by both agents.
\begin{proposition}
\label{prop: 2_player_planner}
The socially optimal outcome is for both players to exploit the safe arm whenever $\pi \leq \underline{\pi}^*$, and to jointly explore the risky arm on $\pi \geq \bar{\pi}^*$, where
\begin{equation*}
\begin{gathered}
\underline{\pi}^*=\frac{\alpha(1-\delta)}{\beta\left[(1+\alpha)(1-\delta)+\delta (\beta-\alpha(1-\beta)) (1+p)\right]}\,, \\
\bar{\pi}^* = \frac{\alpha(1-\delta)}{\beta\left[(1+\alpha)(1-\delta)+\delta(\beta-\alpha(1-\beta))\left(1+p(1-2\beta)\right)\right]}\,.
\end{gathered}
\end{equation*}
\end{proposition}
The substantial lesson from this proposition is that the equilibrium outcome features over-exploitation ($\underline{\pi}>\underline{\pi}^*$) and under-exploration ($\bar{\pi}>\bar{\pi}^*$) relative to the social optimum.
\begin{figure}[h]
\centering
\begin{tikzpicture}[font=\small]
  \def\axisLen{10}
  \def\piUSX{1.5}    % underline pi star  (social optimum exploitation cutoff)
  \def\piUX{2.7}     % underline pi       (equilibrium exploitation cutoff)
  \def\piBSX{4.8}    % bar pi star        (social optimum exploration cutoff)
  \def\piBX{6.8}     % bar pi             (equilibrium exploration cutoff)
  \def\tauX{8.5}     % tau                (second-period exploration threshold)

  % Row geometry
  \def\topBot{0.25}   % bottom of top row
  \def\topTop{0.80}   % top of top row
  \def\botTop{-0.25}  % top of bottom row
  \def\botBot{-0.80}  % bottom of bottom row

  % -----------------------------------------------------------
  \fill[white,draw=black,thin]   (0,\topBot)    rectangle (\piUSX,\topTop);
  \fill[gray!40,draw=black,thin] (\piUSX,\topBot) rectangle (\piBSX,\topTop);
  \fill[gray!80,draw=black,thin] (\piBSX,\topBot) rectangle (\axisLen,\topTop);

  % Equilibrium row: split at piUX and piBX
  \fill[white,draw=black,thin]   (0,\botBot)    rectangle (\piUX,\botTop);
  \fill[gray!40,draw=black,thin] (\piUX,\botBot) rectangle (\piBX,\botTop);
  \fill[gray!80,draw=black,thin] (\piBX,\botBot) rectangle (\axisLen,\botTop);

  % Top row (Social optimum)
  \node[font=\scriptsize] at ({\piUSX/2}, 0.525)                   {Exploit};
  \node[font=\scriptsize] at ({(\piUSX+\piBSX)/2}, 0.525)         {Asymmetric};
  \node[font=\scriptsize,white] at ({(\piBSX+\axisLen)/2}, 0.525) {Explore};
  % Bottom row (Equilibrium)
  \node[font=\scriptsize] at ({\piUX/2}, -0.525)                   {Exploit};
  \node[font=\scriptsize] at ({(\piUX+\piBX)/2}, -0.525)          {Asymmetric};
  \node[font=\scriptsize,white] at ({(\piBX+\axisLen)/2}, -0.525) {Explore};

  % -----------------------------------------------------------
  % OPTIMUM THRESHOLD LABELS above the social optimum row
  % -----------------------------------------------------------
  \node[above=2pt, font=\small] at (\piUSX, \topTop) {$\underline{\pi}^*$};
  \node[above=2pt, font=\small] at (\piBSX, \topTop) {$\bar{\pi}^*$};

  % -----------------------------------------------------------
  % ROW LABELS (left of figure)
  % -----------------------------------------------------------
  \node[left, align=right, font=\scriptsize] at (-0.2,  0.525) {Social\\[-2pt]optimum};
  \node[left, align=right, font=\scriptsize] at (-0.2, -0.525) {Equilibrium};

  % -----------------------------------------------------------
  % BELIEF AXIS (drawn on top of fills)
  % -----------------------------------------------------------
  \draw[thick,->] (0,0) -- (\axisLen+0.5, 0) node[right] {$\pi$};
  \draw[thick] (0,       0.06) -- (0,       -0.06);
  \draw[thick] (\axisLen,0.06) -- (\axisLen,-0.06);

  % -----------------------------------------------------------
  % CUTOFF TICK MARKS on the axis
  % -----------------------------------------------------------
  \draw[thick] (\piUSX, 0.06) -- (\piUSX, -0.06);
  \draw[thick] (\piUX,  0.06) -- (\piUX,  -0.06);
  \draw[thick] (\piBSX, 0.06) -- (\piBSX, -0.06);
  \draw[thick] (\piBX,  0.06) -- (\piBX,  -0.06);
  \draw[thick] (\tauX,  0.06) -- (\tauX,  -0.06);

  % -----------------------------------------------------------
  % CUTOFF LABELS below the bottom row 
  % -----------------------------------------------------------
  \node[below=2pt, font=\small] at (0,        -0.80) {$0$};
  \node[below=2pt, font=\small] at (\piUX,   -0.80) {$\underline{\pi}$};
  \node[below=2pt, font=\small] at (\piBX,   -0.80) {$\bar{\pi}$};
  \node[below=2pt, font=\small] at (\tauX,   -0.80) {$\tau$};
  \node[below=2pt, font=\small] at (\axisLen,-0.80) {$1$};

\end{tikzpicture}
\caption{Equilibrium and optimum thresholds on the unit interval; $\underline{\pi}$ is the equilibrium exploitation cutoff; $\bar \pi$ is the equilibrium exploration cutoff; $\underline{\pi}^*$ is the optimal exploitation cutoff; $\bar{\pi}^*$ is the optimal exploration cutoff; $\tau$ is the second-period exploration cutoff. In equilibrium (bottom), over-exploitation ($\underline{\pi} > \underline{\pi}^*$) and under-exploration ($\bar{\pi} > \bar{\pi}^*$) shift the thresholds rightward relative to the social optimum.}
\label{fig: beliefs_cutoffs}
\end{figure}

We wrap up this section by investigating the effect of the connection probability $p$ on equilibrium social surplus. This is a finite-player pure-strategy comparison. Social surplus is \textit{increasing} in $p$ within each equilibrium region, but drops discontinuously at the regime change from full exploration to asymmetric equilibrium, as illustrated in Figure~\ref{fig: SS} for two levels of initial beliefs $\pi_2 > \pi_1$. Let $p(\pi)$ be the level at which $\bar{\pi}(p)=\pi$. When $p< p(\pi)$, full exploration prevails and social surplus increases with $p$. Once $p$ surpasses $p(\pi)$, the number of explorers drops from two to one, causing the discontinuous fall. Thereafter, raising $p$ increases social surplus by enhancing information sharing without altering free-riding incentives.

\begin{figure}[htp]
    \centering
    % --- Subfigure (a) ---
    \begin{subfigure}{.48\textwidth}
        \centering
        \begin{tikzpicture}
            \tikzmath{
                \Del = 0.9; \Al = 0.5; \Be = 0.7; \P = 0.6;
                \PLowerBarStar = 0.07582;
                \PLowerBar = 0.11074;
                \PUpperBarStar = 0.13574;
                \PUpperBar = 0.16341;
            }
            \begin{axis}[
                smooth,
                axis lines=left,
                axis line style={-stealth},
                xmin=0, xmax=0.22,
                ymin=0, ymax=0.1,
                xtick={\PLowerBarStar, \PLowerBar, \PUpperBarStar, \PUpperBar},
                xticklabels={$\underline{\pi}^*$, $\underline{\pi}$, $\bar{\pi}^*$, $\bar{\pi}$},
                ytick={0},
                yticklabels={0},
                xlabel={$\pi$},
                x label style={at={(axis description cs:1,0)}, anchor=north},
                legend style={at={(0.03,0.97)}, anchor=north west, nodes={scale=0.7, transform shape}, cells={anchor=west}}
            ]
                % Optimal Social Surplus (Dashed)
                \addplot[thick, dashed, domain=0.06:\PLowerBarStar, forget plot] {0};
                \addplot[thick, dashed, samples=400, domain=\PLowerBarStar:\PUpperBarStar] 
                    {((1-\Del)*((x*\Be)-(\Al*(1-(x*\Be)))))+(\Del*x*\Be*(\Be-(\Al*(1-\Be)))*(1+\P))};
                \addplot[thick, dashed, samples=400, domain=\PUpperBarStar:0.2] 
                    {(2*(1-\Del)*((x*\Be)-(\Al*(1-(x*\Be))))) + (2*\Del*x*(\Be-(\Al*(1-\Be)))*\Be*(1+\P*(1-\Be)))};
                
                % Equilibrium Social Surplus (Solid)
                \addplot[ultra thick, domain=0:\PLowerBar, forget plot] {0};
                \addplot[thick, samples=400, domain=\PLowerBar:\PUpperBar] 
                    {((1-\Del)*((x*\Be)-(\Al*(1-(x*\Be)))))+(\Del*x*\Be*(\Be-(\Al*(1-\Be)))*(1+\P))};
                \addplot[thick, samples=400, domain=\PUpperBar:0.2] 
                    {(2*(1-\Del)*((x*\Be)-(\Al*(1-(x*\Be))))) + (2*\Del*x*(\Be-(\Al*(1-\Be)))*\Be*(1+\P*(1-\Be)))};
                
                \legend{optimal social surplus, equilibrium social surplus};
                
                % Markers for jumps
                \addplot[only marks, mark=*] coordinates {(\PLowerBar,0) (\PUpperBar,0.05775)};
                \addplot[fill=white, only marks, mark=*] coordinates {(\PLowerBar,0.023023) (\PUpperBar,0.0679478)};
            \end{axis}
        \end{tikzpicture}
        \caption{Optimal vs. equilibrium social surplus}
    \end{subfigure}
    \hfill
    % --- Subfigure (b) ---
    \begin{subfigure}{.48\textwidth}
        \centering
        \begin{tikzpicture}
            \tikzmath{
                \Del = 0.4; \Al = 0.65; \Be = 0.85;
                \pcuttwo = 0.776091; \pcutone = 0.5625998;
                \pione = 0.4; \pitwo = 0.42;
                \Uotwo = 0.1542834; \Uoone = 0.106516;
                \Uttwo = 0.1667929; \Utone = 0.115152;
            }
            \begin{axis}[
                smooth,
                axis lines=left,
                axis line style={-stealth},
                xmin=0, xmax=1.1,
                ymin=0.02,
                xtick={0, \pcutone, \pcuttwo, 1},
                xticklabels={0, $p(\pi_1)$, $p(\pi_2)$, 1},
                ytick=\empty,
                xlabel={$p$},
                x label style={at={(axis description cs:1,0)}, anchor=north},
                legend style={at={(0.95,0.05)}, anchor=south east, nodes={scale=0.9, transform shape}}
            ]
                % Curve for pi2 (Solid)
                \addplot[thick, samples=400, domain=\pcuttwo:1] 
                    {((1-\Del)*((\pitwo*\Be)-(\Al*(1-(\pitwo*\Be)))))+(\Del*\pitwo*\Be*(\Be-(\Al*(1-\Be)))*(1+x))};
                \addplot[forget plot, thick, samples=400, domain=0:\pcuttwo] 
                    {(2*(1-\Del)*((\pitwo*\Be)-(\Al*(1-(\pitwo*\Be))))) + (2*\Del*\pitwo*(\Be-(\Al*(1-\Be)))*\Be*(1+x*(1-\Be)))};
                
                % Curve for pi1 (Dashed)
                \addplot[thick, dashed, samples=400, domain=\pcutone:1] 
                    {((1-\Del)*((\pione*\Be)-(\Al*(1-(\pione*\Be)))))+(\Del*\pione*\Be*(\Be-(\Al*(1-\Be)))*(1+x))};
                \addplot[forget plot, thick, dashed, samples=400, domain=0:\pcutone] 
                    {(2*(1-\Del)*((\pione*\Be)-(\Al*(1-(\pione*\Be))))) + (2*\Del*\pione*(\Be-(\Al*(1-\Be)))*\Be*(1+x*(1-\Be)))};
                
                % Marks and vertical dotted lines
                \draw [dotted] (axis cs:\pcutone, 0.02) -- (axis cs:\pcutone, \Utone);
                \draw [dotted] (axis cs:\pcuttwo, 0.02) -- (axis cs:\pcuttwo, \Uttwo);
                
                \addplot[only marks, mark=*, forget plot] coordinates {(\pcutone,\Uoone) (\pcuttwo,\Uotwo)};
                \addplot[only marks, mark=*, fill=white, forget plot] coordinates {(\pcutone,\Utone) (\pcuttwo,\Uttwo)};
                
                \legend{$\pi = \pi_2$, $\pi = \pi_1$};
            \end{axis}
        \end{tikzpicture}
        \caption{Equilibrium social surplus with respect to the connection probability $p$ for levels of initial belief $\pi_1 < \pi_2$}
    \end{subfigure}

    \vspace{1em}
    \caption{Social surplus}
    \label{fig: SS}
\end{figure}

In Section~\ref{subs: equil_soc_surplus}, we show that this finite-player pure-strategy pattern---increasing surplus within each regime, with drops when the integer number of explorers falls---persists with many players. We then contrast it with the finite-player symmetric mixed-strategy benchmark, where expected per-capita surplus varies continuously with connectivity, and the large-population welfare benchmark, where the discontinuities vanish.
%%%%%%%%%%%%%%%%%%%%%%%%%%%%%%%%%%%%%%%%%%%%%%%%%%%%%%%%%%%%%%%%%%%%%%%%%%%%%%%%%%%%%%%%%%%%%%%
\section{Equilibria in the \texorpdfstring{$n$}{n}-Agent Economy}
\label{sec: many_players_Econ}
We study an economy with $n$ ex ante identical agents who share a common initial belief $\pi$. In the second period, each agent observes the exploration outcomes of a randomly selected group of size $M$. The realization of this group is unknown in the first period; agents know only the distribution of $M$, which is common across all players.

Aside from the above, we impose no further structure on the connection graph; our results depend only on the probability-generating function of $M$. We characterize pure-strategy equilibria by $\pi$: full exploitation at low beliefs, an asymmetric equilibrium with both explorers and exploiters at intermediate beliefs, and full exploration at high beliefs. In the intermediate region, we also identify the unique symmetric mixed-strategy equilibrium, our focal selection, and study its comparative statics. These results provide the three objects used later: the full-exploration threshold, the pure-strategy number of explorers, and the symmetric mixed exploration probability.

In Section~\ref{sec: limit}, we apply these results to Erd\H{o}s-R\'enyi random graphs and analyze the limiting equilibrium properties as $n \to \infty$.
%------------------------------------------------
\subsection{Symmetric Pure-Strategy Equilibria}
\label{subs: symm_pure_strat}
We study two symmetric pure-strategy equilibria: exploitation and exploration. The exploitation equilibrium occurs when \textit{all} players choose the safe arm in the first period, and prevails whenever the initial belief falls below the threshold $\underline{\pi}$ in equation~\eqref{eq: full_sharing_thresholds}. One can verify this by comparing an individual's exploitation payoff when everyone else exploits ($w_0(\pi)$) with her exploration payoff as the sole explorer ($v_1(\pi)$). This implies $w_0(\pi) \geq v_1(\pi)$ whenever $\pi \leq \underline{\pi}$. Thus, the exploitation equilibrium condition remains unchanged despite having more than two players and imperfect connections.

The more intriguing case is the examination of the existence of the exploration equilibrium in which \textit{all} agents choose the risky arm in the first period. 

Define $v_k(\pi)$ and $w_k(\pi)$ as an agent’s expected payoff from exploration and exploitation in the first period when, in total, $k$ agents explore. Thus the exploration equilibrium appears if and only if $v_n(\pi) > w_{n-1}(\pi)$. We reduce this inequality and prove the next theorem in the appendix.
\begin{theorem}[Full-exploration cutoff under connectivity shifts]
\label{thm: exploration_equil}
Let $M$ be the size of the random group of contacts in the second period. Then, the exploration equilibrium appears on $\pi > \bar{\pi}$, where
\begin{equation}
\label{eq: n_player_pi_bar}
    \bar{\pi} = \frac{\alpha (1-\delta)}{\beta\left((1+\alpha)(1-\delta)+\delta \left(\beta-\alpha(1-\beta)\right) \BE\left[(1-\beta)^{M}\right]\right)}\,.
\end{equation}
\end{theorem}
Our key comparative static studies how connection sparsity shifts the exploration threshold. Below, we consider (\rn{1}) first-order stochastic dominance (\textsf{FOSD}) shifts and (\rn{2}) mean-preserving spreads (\textsf{SOSD}) of the distribution of $M$.

\paragraph{FOSD shift.} Since $x\mapsto (1-\beta)^x$ is decreasing, a \textsf{FOSD} increase in the distribution of $M$ lowers $\BE\left[(1-\beta)^M\right]$. Hence the exploration threshold rises (the exploration region shrinks): this means denser connections strengthen free-riding.

\paragraph{SOSD shift.} Fix $\beta\in(0,1)$ and let $M'' \succeq_2 M'$ (i.e., $M''$ second-order stochastically dominates $M'$). By definition of \textsf{SOSD}, for every nondecreasing and concave function $f$, one has
\[
\BE\left[f(M'')\right] \geq \BE\left[f(M')\right].
\]
Take $f(x) \coloneq -(1-\beta)^x$. Since $1-\beta \in (0,1)$,
\[
f'(x) = -\ln(1-\beta)\,(1-\beta)^x > 0
\quad\text{and}\quad
f''(x) = -\big(\ln(1-\beta)\big)^2 (1-\beta)^x < 0,
\]
so $f$ is nondecreasing and concave, implying that,
\[
\BE\!\left[(1-\beta)^{M''}\right] \le \BE\!\left[(1-\beta)^{M'}\right].
\]
Therefore, if $\widetilde{M}$ is a mean-preserving spread of $M$ (so that $M \succeq_2 \widetilde{M}$), then $\bar{\pi}(\widetilde{M}) \le \bar{\pi}(M)$. In other words, greater second-order risk (a mean-preserving spread) lowers the exploration threshold, making full first-period exploration more attractive.

%--------------------------------------------------------------------
%-------------------------------------------------------------------------------
\subsection{Asymmetric Pure-Strategy Equilibria (Intermediate Region)}
\label{subs: asymm_pure_strat}
Previously, we characterized symmetric equilibria with full exploration or full exploitation. Here we study the intermediate region $\pi\in(\underline{\pi},\bar\pi]$, where a pure-strategy equilibrium can feature both explorers and exploiters. Let $k\in\{1,\dots,n-1\}$ be the equilibrium number of explorers. Such an equilibrium requires that (\rn{1}) explorers prefer to keep exploring rather than deviate to exploitation, i.e., $v_k(\pi)>w_{k-1}(\pi)$, and (\rn{2}) exploiters prefer to keep exploiting rather than deviate to exploration, i.e., $w_k(\pi)\ge v_{k+1}(\pi)$. 
As a first step toward analyzing such equilibria, we show that for large values of $\pi$, the second incentive constraint (i.e., $w_k(\pi)\ge v_{k+1}(\pi)$) fails.
\begin{lemma}
\label{lem: no_exploiting_large_pi}
If $\pi>\tau$, then $w_k(\pi)<v_{k+1}(\pi)$ for all $k$.
\end{lemma}
This lemma implies that a pure-strategy equilibrium with a nonzero number of exploiters cannot exist when $\pi>\tau$. Hence, to characterize intermediate equilibria, we restrict attention to $\pi\le\tau$. The next theorem provides the belief cutoffs supporting an equilibrium with $k$ explorers.
\begin{theorem}[Asymmetric pure-strategy equilibrium]
\label{thm: assymetric_equil}
An asymmetric equilibrium with $k\in\{1,\dots,n-1\}$ explorers exists if and only if $\pi\in(\pi_{k-1,n},\pi_{k,n}]$, where the cutoffs are given by
\begin{equation}
\label{eq: asymmetric_equil_cond}
\begin{gathered}
\pi_{k, n} \coloneq \frac{\alpha(1-\delta)}{\beta\left((1+\alpha)(1-\delta)+\delta (\beta-\alpha(1-\beta)) \BE_k\left[(1-\beta)^M\right]\right)}\,.
\end{gathered}
\end{equation}
\end{theorem}
The lower cutoff ($\pi>\pi_{k-1,n}$) follows from the explorers' incentive constraint $v_k>w_{k-1}$, while the upper cutoff ($\pi\le\pi_{k,n}$) follows from the exploiters' constraint $w_k\ge v_{k+1}$. Deriving these conditions parallels the steps in Theorem~\ref{thm: exploration_equil}, so we omit the proof.

As a result, the asymmetric equilibrium with $k$ explorers occurs whenever $\pi_{k-1,n}<\pi\le \pi_{k,n}$. Full exploitation prevails for $\pi\le \pi_{0,n}\equiv \underline\pi$, and full exploration for $\pi>\pi_{n-1,n}\equiv \bar\pi_n$.\footnote{We use the subscript $n$ under the exploration threshold $\bar\pi$ when we want to be explicit about the total number of players.}

\paragraph{Comparative statics of equilibrium $k$ in $\pi$.} Let $M^{(n)}_{k}$ denote the number of second-period contacts when $n$ agents are present and $k$ explore in period~1.\footnote{\label{foot: exp_notation}Depending on the context, we either use $M^{(n)}_k$ or explicitly specify the indices on the expectation operator, that is e.g., $\BE^{(n)}_k$.} A simple stochastic dominance argument shows $M^{(n)}_{k}\succeq_1 M^{(n)}_{k-1}$, hence $\pi_{k-1,n}\le \pi_{k,n}$. Therefore, the equilibrium number of explorers ($k$) is weakly increasing in $\pi$.\footnote{The term `weakly' is used because the equilibrium number of explorers is constant on each interval $(\pi_{k-1,n},\pi_{k,n}]$.}

%---------------------------------------------

\subsection{Symmetric Mixed-Strategy Equilibrium (Intermediate Region)}
\label{subs: mixed_strat}
In the previous two sections, we characterized pure-strategy equilibria as a function of the initial belief $\pi$. Section~\ref{subs: symm_pure_strat} shows that at the two ends, i.e., on $[0,\underline\pi]$ and $(\bar\pi,1]$, the unique equilibrium is full exploitation and full exploration, respectively. Section~\ref{subs: asymm_pure_strat} implies that on the intermediate region $(\underline\pi,\bar\pi]$, the unique pure-strategy equilibrium is asymmetric, featuring both explorers and exploiters. This leaves open the possibility of mixed-strategy equilibria on the intermediate region.

Hence in this section, we study the \emph{symmetric} mixed-strategy equilibrium in the intermediate region. Suppose that each agent explores the risky arm with probability $\mu$. This will be a mixed-strategy equilibrium if the expected payoff from exploitation, namely
\begin{equation*}
    w(\pi;\mu) \coloneq \sum_{k=0}^{n-1} \binom{n-1}{k}\mu^k(1-\mu)^{n-1-k}w_k(\pi)\,,
\end{equation*}
matches the expected payoff from exploration, that is
\begin{equation*}
    v(\pi;\mu) \coloneq \sum_{k=0}^{n-1} \binom{n-1}{k}\mu^k(1-\mu)^{n-1-k}v_{k+1}(\pi)\,.
\end{equation*}
\begin{theorem}[Symmetric mixed-strategy equilibrium]
\label{thm: symmetric_mixed_equil}
For $\pi \in \left(\underline{\pi}, \bar{\pi}\right]$ (the intermediate region), there exists a unique symmetric mixed-strategy equilibrium. Moreover, the equilibrium probability of exploration, $\mu^e$, is increasing in $\pi$ and decreasing in $n$.
\end{theorem}
The theorem has two key implications. First, it implies that the symmetric mixed-strategy equilibrium is unique. Second, it implies that the equilibrium exploration probability decreases with the number of agents, consistent with free-riding incentives in the mixed-strategy setting.

%-------------------------------------------------------------------------------
\section{Large-\texorpdfstring{$n$}{n} Limits of Equilibria in Erd\H{o}s-R\'enyi Economies}
\label{sec: limit}
We now put the connectivity comparative static into a one-parameter Erd\H{o}s-R\'enyi (ER) random-graph environment. Each pair of agents is connected with probability $p=\lambda/n$, where $\lambda$ is the average number of immediate neighbors.

The random group of contacts observed in period two can be (\rn{1}) immediate neighbors (the \textbf{local} case), or (\rn{2}) all agents in the connected component (the \textbf{global} case). In the global case, an agent observes signals from everyone in her component $\mathcal{C}$, with $|\mathcal{C}|=M+1$.\footnote{Each player’s connected component includes herself; hence in the \emph{latter} case, $M$ denotes the number of \emph{other} players connected to the focal agent.} In the global case, we interpret period two as a sequence of message-passing sub-periods along paths in the graph. Holding the prior belief fixed, the local economy yields a closed-form decline in experimentation as $\lambda$ rises; the global economy shows how this decline changes sharply around the Erd\H{o}s-R\'enyi phase transition.

We first study local connections, where $M=D$ (the degree), and then global connections, where $M=|\mathcal{C}|-1$ (component size minus one). In both cases we focus on limits as $n\to\infty$ and the effect of $\lambda$ on equilibrium outcomes.
%-------------------------------------------

\subsection{Local Connections and the Share of Explorers}
\label{subs: local_connect}
Recall that in the local regime $M=D$, the degree of a randomly drawn agent follows a Binomial distribution $\mathsf{Bin}(n-1,p)$. For constant $\lambda$, this distribution converges weakly to $\mathsf{Poisson}(\lambda)$, and therefore the limiting threshold for the full exploration equilibrium in the local regime is
\begin{equation}
\label{eq: local_exploration_threshold_limit}
\begin{aligned}
   \bar{\pi}^{\text{local}}_\infty& \coloneq \lim_{n \to \infty} \bar{\pi}^{\text{local}}_n =\lim_{n \to \infty} \frac{\alpha (1-\delta)}{\beta\left((1-\delta)(1+\alpha)+\delta \left(\beta-\alpha(1-\beta)\right) \BE\left[(1-\beta)^{D}\right]\right)}\\
   &=\frac{\alpha(1-\delta)}{\beta\left((1-\delta)(1+\alpha)+\delta \left(\beta-\alpha(1-\beta)\right) e^{-\lambda \beta}\right)}\,,
\end{aligned}
\end{equation}
where the last equality uses the probability-generating function of the Poisson distribution. Observe that as $\lambda \to \infty$, $\bar{\pi}^{\text{local}} \uparrow \tau$.
\begin{lemma}
\label{lem: monotonicity_local_explr_threshold}
In the local regime, the exploration threshold $\bar{\pi}_n$ converges to $\bar{\pi}_\infty^{\mathrm{local}}$ defined in~\eqref{eq: local_exploration_threshold_limit}.
\end{lemma}

In the next lemma, we study the exploration cutoffs in the intermediate region that characterize the asymmetric equilibria (see Theorem~\ref{thm: assymetric_equil}).
\begin{lemma}
\label{lem: ordering}
In the local regime, for a fixed $k\in \BN$ and sufficiently large $n$, the following ordering holds: $\pi_{k-1,n}\leq \pi_{k,n+1}\leq \pi_{k,n} \leq \pi_{k+1,n+1}$.
\end{lemma}
This lemma demonstrates that for large economies with local connections, adding one more individual \textit{never} results in fewer exploring agents in the equilibrium. Put differently, the existing agents do not switch to the exploitation status when new individuals join the economy.\footnote{It is important to note that this conclusion primarily rests on keeping the average degree $\lambda$ constant while expanding the size of the economy.}

Let $k_n$ denote the equilibrium number of exploring agents under pure strategies. The following proposition shows that in an ER economy with local connections, the equilibrium fraction of explorers $k_n/n$ converges as $n \to \infty$. The proof uses the incentive condition~\eqref{eq: asymmetric_equil_cond} to establish matching upper and lower bounds for $k_n$.

For convenience and to avoid repetition, we use the following definition.
\begin{definition}
Define the functions $C_1(\cdot)$ and $C_2(\cdot)$ as
\begin{equation*}
    C_1(\pi)\coloneq-(1-\delta)\big(\pi\beta-\alpha(1-\pi\beta)\big)\,,\qquad C_2(\pi)\coloneq\delta\pi(\beta-\alpha(1-\beta))\beta\,,
\end{equation*}
both strictly positive on the relevant interior domain $\pi\in[\underline{\pi},\tau)$.
\end{definition}
\begin{proposition}[Connectivity and the limiting fraction of explorers in pure equilibria]
\label{prop: equil_num_explr_agents}
Let $k_n(\pi)$ be the equilibrium number of exploring agents in an Erd\H{o}s-R\'enyi economy of $n$ individuals with local connections. Then
\begin{equation}
\label{eq: kappa_expression}
    \lim_{n \to \infty} \frac{k_n(\pi)}{n}=\kappa(\pi) \coloneq
    \left\{
    \begin{array}{ll}
        0 & \pi \leq \underline{\pi}\\
        \frac{1}{\lambda \beta}\log\frac{C_2(\pi)}{C_1(\pi)} &  \underline{\pi}< \pi < \bar{\pi}^{\mathrm{local}}_\infty \\
        1 & \pi \geq \bar{\pi}^{\mathrm{local}}_\infty
	    \end{array} \right.
\end{equation}
At a fixed prior in the intermediate region, $\kappa(\pi)$ is strictly decreasing in $\lambda$ while that prior remains in the region.
\end{proposition}
\begin{figure}[htbp]
\centering
\begin{tikzpicture}[scale=0.9]
\tikzmath{
\Del = 0.35;
\Al = 0.5;
\Be = 0.8;
\Lam = 3;
\PLowerBar = 0.332991;
\PInftylocal = 0.407380;
}
\begin{axis}[smooth,
    axis lines=left,
    xmin=0.31,
    xmax=0.43,
    ymax=1.1,
	axis line style={-stealth},
 	xtick={\PLowerBar,\PInftylocal},
 	xticklabels={$\underline{\pi}$,$\bar{\pi}_\infty^{\text{local}}$},
	ytick={0,1},
	x label style={at={(axis description cs:1,0)},anchor=north},
	y label style={at={(axis description cs:0,1)},anchor=east,rotate=-90},
	xlabel={$\pi$},
	ylabel={$\kappa(\pi)$},
    ]
    \addplot[thick,samples=400, domain = \PLowerBar:\PInftylocal ] {(ln((\Del*x*\Be*(\Be-(\Al*(1-\Be))))/((1-\Del)*(\Al*(1-(x*\Be))-(x*\Be)))))/(\Lam*\Be)};
    
    \addplot[dotted, thick] coordinates {(\PInftylocal,0) (\PInftylocal,1)};
    \addplot[thick,samples=400, domain = \PInftylocal:1] {1};
    
    \addplot[thick,samples=400, domain = 0:\PLowerBar] {0};
\end{axis}   
\end{tikzpicture}
\caption{Limiting fraction of explorers}
\label{fig: limiting_fraction}
\end{figure}
Figure~\ref{fig: limiting_fraction} illustrates the limiting fraction of exploring agents, $\kappa(\pi)$, as a function of the initial belief $\pi$. The function exhibits two kinks at $\underline{\pi}$ and $\bar{\pi}_\infty^{\text{local}}$, corresponding to equilibrium regime changes from full exploitation to the intermediate asymmetric region and then to full exploration. Notably, the graph is convex, indicating that the equilibrium fraction of explorers grows at an increasing rate as the initial belief rises. For the connectivity comparative static emphasized in this paper, the key point is the dependence on $\lambda$: on the intermediate region, $C_1(\pi)$ and $C_2(\pi)$ are fixed by primitives and
\[
    \kappa(\pi)=\frac{1}{\lambda\beta}\log\frac{C_2(\pi)}{C_1(\pi)}
\]
is decreasing in $\lambda$. Thus local connectivity reduces the share of agents producing information.

We next characterize the finite-$n$ equilibrium exploration probability $\mu^e_n$ in the intermediate region $[\underline{\pi},\bar{\pi}_n^{\mathrm{local}}]$ where the symmetric equilibrium is mixed (Theorem~\ref{thm: symmetric_mixed_equil}), and then find its limit as $n\to\infty$. 
\begin{proposition}[Connectivity and mixed equilibria]
\label{prop:equil_mix_prob}
For finite $n$, $0<\lambda\leq n$, and every $\pi\in[\underline{\pi},\bar{\pi}_n^{\mathrm{local}}]$, the equilibrium mixing probability is
\begin{equation}
\label{eq:equil_mix_prob}
    \mu^e_n = \frac{n}{\lambda\beta}\left[1-\left(\frac{C_1(\pi)}{C_2(\pi)}\right)^{1/(n-1)}\right],
\end{equation}
with limit
\begin{equation}
\label{eq:limit_equil_mix_prob}
    \lim_{n\to\infty}\mu^e_n = \frac{1}{\lambda\beta}\log\frac{C_2(\pi)}{C_1(\pi)}.
\end{equation}
Holding $\pi$ fixed in the interior mixing region, $\mu^e_n$ is strictly decreasing in $\lambda$.
\end{proposition}
This proposition shows, interestingly, that the limiting exploration probability along a sequence of mixed-strategy equilibria equals the asymptotic fraction of explorers along a sequence of pure-strategy equilibria---an instance of the \emph{exact law of large numbers} in the sense of~\cite{sun2006exact}.
\begin{corollary}[Constant expected utility under interior symmetric mixing]
\label{cor:mixed_expected_utility}
Fix $n$ and $\pi$. For any $0<\lambda\leq n$ for which $\pi$ lies in the interior symmetric-mixing region, the expected payoff of each agent, and hence the expected per-capita equilibrium social surplus, is
\begin{equation*}
    \begin{aligned}
    S_n(\pi,\lambda)
    &=\delta\pi\bigl(\beta-\alpha(1-\beta)\bigr)
    \left(1-\frac{C_1(\pi)}{C_2(\pi)}\right)\\
    &=\frac{1-\delta}{\beta}\bigl(\pi\beta-\alpha(1-\pi\beta)\bigr)
    +\delta\pi\bigl(\beta-\alpha(1-\beta)\bigr),
    \end{aligned}
\end{equation*}
and is therefore independent of $\lambda$.
\end{corollary}
Higher connectivity makes an agent more likely to learn from others, but the equilibrium exploration probability falls by exactly enough to keep the probability that the agent receives no successful signal from others unchanged. The payoff from exploitation is therefore constant, and indifference implies that exploration yields the same expected payoff. Corollary~\ref{cor:mixed_expected_utility} thus provides the finite-$n$ symmetric-mixing welfare benchmark: within the interior mixing region, connectivity changes the balance between information production and diffusion without changing expected per-capita surplus.
%------------------------------------------

\subsection{Global Connections and the Phase Transition} 
\label{subs: global_conn}
The analysis in the global regime (where $M = |\mathcal{C}|-1$) is more intricate. In this regime, an agent observes all members of her connected component in the second period. This setting offers greater potential for knowledge diffusion due to the increased number of connections, but also creates stronger free-riding incentives, leading to lower levels of exploration.

A coupling argument (e.g., Theorem~2.1 in~\cite{bollobas2001random}) shows that the distribution of the connected component size $|\mathcal{C}|$ of a \emph{typical} vertex in an ER random graph is first-order stochastically increasing in $\lambda$. Since $x\mapsto (1-\beta)^x$ is decreasing, equation~\eqref{eq: n_player_pi_bar} implies that the exploration threshold $\bar{\pi}$ is increasing in $\lambda$, confirming the presence of free-riding in the global case.

To study the limiting behavior of the exploration threshold, we need the asymptotic distribution of $|\mathcal{C}|$, the size of a typical component. A classical result in random graph theory establishes that the giant component undergoes a phase transition: its size is $O(\log n)$ in the subcritical regime ($\lambda < 1$) and $O(n)$ in the supercritical regime ($\lambda > 1$). This phenomenon in turn impacts the distribution of the size of a \emph{typical} (not necessarily giant) component containing a uniformly chosen vertex.

Let $B(\lambda)$ denote the \emph{total} number of descendants in a branching process with $\mathsf{Poisson}(\lambda)$ offspring distribution. In the (sub)critical regime ($\lambda \leq 1$), its distribution is known as the $\mathsf{Borel}(\lambda)$ distribution:
\begin{equation*}
    \BP\left(B(\lambda)=k\right) = \frac{\mathrm{e}^{-\lambda k}(\lambda k)^{k-1}}{k!}
\end{equation*}
In the supercritical regime ($\lambda > 1$), the total population becomes extinct with probability $\zeta(\lambda)$, which is the unique solution in $(0,1)$ to the following fixed-point relation:
\begin{equation}
\label{eq: extinct_prob}
    \zeta = \mathrm{e}^{-\lambda (1-\zeta)}
\end{equation}
In the supercritical regime, a uniformly random vertex falls in the giant (non-extinct) component with probability $1-\zeta$, and hence the asymptotic size of its component is $\infty$. With probability $\zeta$, it belongs to a finite-size extinct component whose size follows $\mathsf{Borel}(\lambda \zeta)$. We use the following notation to refer to this mixture distribution:
\begin{equation*}
    \mathsf{Borel}_\zeta (\lambda) \coloneq \left\{ 
    \begin{array}{lc}
        \infty & \text{with probability } 1-\zeta \\
        \mathsf{Borel}(\lambda \zeta) & \text{with probability } \zeta
    \end{array}\right.
\end{equation*}
We can further define $\zeta(\lambda) \equiv 1$ for all $\lambda \leq 1$, so that the standard Borel distribution can be represented as a special case of the mixture Borel: $\mathsf{Borel}(\lambda) = \mathsf{Borel}_{\zeta(\lambda)} (\lambda)$. We denote a random variable with the standard Borel distribution by $B(\lambda)$, and one with the mixture Borel distribution with extinction probability $\zeta$ by $B_\zeta(\lambda)$.

Using these observations, we completely characterize the asymptotic distribution of the size of a typical connected component $|\mathcal{C}|$ in the following proposition:
\begin{proposition}
\label{prop: asymptotic_pi_bar}
Let $p =\lambda/n$, and denote the size of a typical connected component by $|\mathcal{C}|$. Then:
\begin{enumerate}[label=(\roman*)]
    \item  \label{item: weak_conv} $|\mathcal{C}|$ converges in distribution to $\mathsf{Borel}_{\zeta(\lambda)} (\lambda)$.
    \item \label{item: asymp_pi_bar} As $n\to \infty$, the full exploration threshold converges to
    \begin{equation}
    \label{eq: asymp_pi_bar}
        \bar{\pi}_\infty^{\mathrm{global}}\coloneq\lim_{n \to \infty}\bar{\pi}^\mathrm{global}_n = \frac{\alpha (1-\delta)}{\beta\left((1+\alpha)(1-\delta)+\delta \left(\beta-\alpha(1-\beta)\right) \BE\left[(1-\beta)^{B_{\zeta(\lambda)} (\lambda)-1}\right]\right)}\,.
    \end{equation}
\end{enumerate}
\end{proposition}
To further analyze $\bar{\pi}_\infty^{\mathrm{global}}$ derived from equation~\eqref{eq: asymp_pi_bar}, we need to determine the probability-generating function (henceforth, PGF) of the mixture Borel distribution $B_{\zeta(\lambda)} (\lambda)$. In the following lemma, we characterize that through a fixed-point argument.
\begin{lemma}[\cite{Alon2000}]
\label{lem: mixture_borel_PGF}
    Denote the probability-generating function (PGF) of the mixture Borel distribution by $\psi_\lambda(z) \coloneq \BE\left[z^{B_{\zeta(\lambda)} (\lambda)}\right]$ for $z \in [0,1)$, where $\zeta(\lambda)$ solves~\eqref{eq: extinct_prob}. Then $\psi_\lambda(z)$ satisfies the following fixed-point relation:
    \begin{equation}
        \label{eq: mixture_borel_PGF}
        \psi_\lambda(z) = z \mathrm{e}^{\lambda (\psi_\lambda(z)-1)}
    \end{equation}
\end{lemma}

\paragraph{Rapid fall of exploration in the global regime (small $\beta$ and $\lambda \approx 1$).} As Figure~\ref{fig: rapid_rise} illustrates, there is a rapid tightening of the exploration region in the case of global connections as $\lambda$ increases from values slightly below $1$ to values just above. More precisely, the marginal effect of increasing $\lambda$ on the exploration threshold $\bar\pi^{\text{global}}_\infty$ undergoes a significant shift at $\lambda=1$. This effect is more significant when $\beta$ is close to zero, which is the most relevant region in the innovation and entrepreneurship research, when the probability of success is extremely small. 

\begin{figure}[htbp]
\centering
% \pgfplotsset{normalsize}
\begin{tikzpicture}[]
\begin{axis}[smooth,
    axis line style={thick},
    x label style={at={(axis description cs:1,0)},anchor=south},
    xtick={1, 2, 3, 4},
    xticklabels={1, 2, 3, 4},
    yticklabels={},
    ytick=\empty,
    xmax=5,
    axis lines=left,
    xlabel={$\lambda$},
    ylabel={$\bar\pi^{\text{global}}_\infty$},
    x label style={below},
    ]
    \addplot[line width=1pt, thick] table[col sep=comma] {csv/pi_infty_global_results.csv};
\end{axis}   
\end{tikzpicture}
\captionsetup{width=.5\linewidth}
\caption{Rapid tightening of the exploration region
$[\delta=0.1, \alpha=0.001,\beta=0.002]$}
\label{fig: rapid_rise}
\end{figure}

The emergence of a giant connected component for $\lambda>1$ is a classic result in random graph theory \citep{erd6s1960evolution}. In our setting, this connects the focal agent to a component of size $\Theta(n)$, expands the scope for free-riding, and raises the global exploration threshold $\bar\pi^{\mathrm{global}}_\infty$. For $\lambda<1$, components remain small---$O(\log n)$---so free-riding is limited and the threshold stays low. The key implication is that $\bar\pi^{\mathrm{global}}_\infty$ changes sharply near $\lambda=1$: as $\beta\to0$, the change in the threshold is concentrated in a critical window of width $\Theta(\sqrt{\beta})$ and magnitude $\Theta(\beta^{3/2})$. The next theorem formalizes this scaling and shows that, after rescaling, the threshold curves collapse onto a single explicit ``crossover function''.

\begin{theorem}[Scaling limits]
\label{thm:scaling}
        Let $\beta \to 0$, and suppose
\[
\tau(\alpha,\beta)=\frac{\alpha}{\beta(1+\alpha)} \to \tau_0 \in (0,1).
\]
Then for every fixed $x \in \mathbb{R}$:
\begin{equation*}
\lim_{\beta\downarrow0}\ \beta^{-3/2}\Big(\bar{\pi}_\infty^{\mathrm{global}}(1+x\sqrt\beta,\beta)-\bar{\pi}_\infty^{\mathrm{global}}(1,\beta)\Big)
=\frac{\tau_0(1-\tau_0)\delta}{1-\delta}\,\big(c(x)-\sqrt2\big)\,,
\end{equation*}
where $c(x)\coloneq x + \sqrt{x^2+2}$.
\end{theorem}
\begin{figure}[htbp]
\begin{center}
\pgfplotsset{normalsize}
\begin{tikzpicture}[scale=0.9]
\tikzmath{
\Umin = -6;
\Umax = 6;
}
\begin{axis}[smooth,
    axis lines=left,
    xmin=\Umin,
    xmax=\Umax,
    ymin=0,
    ymax=13,
    axis line style={-stealth},
    xtick={-6,-3,0,3},
    ytick={0,5,10},
    x label style={at={(axis description cs:1,0)},anchor=north},
    y label style={at={(axis description cs:0,1)},anchor=east,rotate=-90},
    xlabel={$x$},
    ylabel={$c(x)$},
]
    % c(u) = u + sqrt(u^2 + 2)
    \addplot[thick, samples=400, domain=\Umin:\Umax] {x + sqrt(x^2 + 2)};

    % reference line at u = 0
    \addplot[dotted, thick] coordinates {(0,0) (0,13)};

    % Optional: asymptote guide for u -> +infty: c(u) ~ 2u
    % (only plotted on the right half for readability)
    \addplot[dashed, thick, samples=200, domain=0:\Umax] {2*x};

\end{axis}
\end{tikzpicture}
\end{center}
\caption{Crossover function \(c(x)=x+\sqrt{x^2+2}\).}
\label{fig:crossover}
\end{figure}
Figure~\ref{fig:crossover} plots the crossover function $c(x)=x+\sqrt{x^{2}+2}$, which governs the behavior of $\bar\pi^{\mathrm{global}}_\infty$ near the phase transition of ER random graphs. The theorem identifies a critical window $\lambda=1+x\sqrt{\beta}$ within which $\bar\pi^{\mathrm{global}}_\infty$ shifts by $\Theta(\beta^{3/2})$; after rescaling by $\beta^{-3/2}$, the entire family of threshold curves collapses onto the closed-form profile $c(x)-\sqrt{2}$.

Outside this window, the PGF term driving the denominator of $\bar\pi^{\mathrm{global}}_\infty$ is essentially stable. For $\lambda<1$, the absence of a giant component means extra links barely affect the component-size distribution, so the threshold is nearly flat---consistent with $c(x)\to 0$ as $x\to-\infty$. For $\lambda>1$, a giant component already exists and further increases in $\lambda$ affect the PGF only gradually, yielding the approximately linear right tail $c(x)\sim 2x$ as $x\to+\infty$.

The crossover function is \emph{asymmetric} around $x=0$: it flattens quickly for $x<0$ but continues to rise for $x>0$. This reflects the underlying phase transition---crossing $\lambda=1$ upward creates access to a giant component and sharply expands free-riding incentives, while moving further below $\lambda=1$ produces no analogous jump, since component sizes remain logarithmic in $n$.
%%%%%%%%%%%%%%%%%%%%%%%%%%%%%%%%%%%%%%%%%%%%%%%%%%%%%%%%%%%%%%%%%%%%%%%%%%%
%%%%%%%%%%%%%%%%%%%%%%%%%%%%%%%%%%%%%%%%%%%%%%%%%%%%%%%%%%%%%%%%%%%%%%%%%%%

\section{Social Surplus in Local Erd\H{o}s-R\'enyi Economies}
\label{sec: social_surplus}
We examine social surplus in the local-connection economy of Section~\ref{subs: local_connect}. We first record the benchmark comparisons. In finite economies, the welfare benchmark is the symmetric mixed equilibrium: expected per-capita equilibrium utility in the interior mixing region is constant in $\lambda$ (Corollary~\ref{cor:mixed_expected_utility}). In the large-$n$ local Erd\H{o}s-R\'enyi limit, per-capita equilibrium surplus under pure strategies and expected per-capita surplus under symmetric mixed strategies converge to the same limit (Proposition~\ref{prop: asympt_equil_soc_surplus} and Theorem~\ref{thm:ubar_mixed_infty}). This common limit is weakly increasing in $\lambda$ and eventually flat for intermediate beliefs. The finite-$n$ pure-strategy equilibrium behaves differently because the number of explorers is an integer: when connectivity crosses a regime-change threshold, the equilibrium may lose one explorer and surplus may drop. We present the finite-$n$ surplus representation first because the limiting results build on it, and then formulate the planner's problem and compare equilibrium experimentation with the social optimum.

Suppose that out of $n$ players, $k$ agents choose the risky arm in the first period, and let the resulting social surplus be denoted as $u_{k,n}(\pi)$. Further, in the ER graphs with local connections, let $q_a(b) = \binom{a}{b}p^b(1-p)^{a-b}$ represent the probability of meeting $b$ agents out of a specific set of $a$ individuals in the second period, then
\begin{small}
\begin{equation}
\label{eq: social_welfare_func}
\begin{gathered}
    u_{k,n}(\pi) = (1-\delta) k\left(\pi \beta-\alpha(1-\pi\beta)\right)\\
    + \delta k \pi (\beta-\alpha(1-\beta)) \beta+\delta k \sum_{m=0}^{k-1}q_{k-1}(m)\left[\pi (\beta-\alpha(1-\beta))(1-\beta)^{m+1}-\alpha(1-\pi)\right]^+ \\
    +\delta k \sum_{m=0}^{k-1} q_{k-1}(m)\pi(\beta-\alpha(1-\beta))(1-\beta)\left(1-(1-\beta)^m\right)\\
    +\delta (n-k) \sum_{m=0}^k q_k(m)\left[\pi (\beta-\alpha(1-\beta))(1-\beta)^m-\alpha(1-\pi)\right]^+ \\ 
    +\delta(n-k) \sum_{m=0}^k q_k(m) \pi (\beta-\alpha(1-\beta))\left(1-(1-\beta)^m\right).
\end{gathered}
\end{equation}
\end{small}
The first line of $u_{k,n}$ is the first-period payoff earned by the $k$ exploring agents. The next two lines give their discounted second-period payoff, which has three components: the expected payoff when an agent herself received a conclusive signal in the first period; the expected payoff when neither she nor any of her second-period contacts did; and the expected payoff when she did not receive a high output but at least one contact did. The last two lines give the discounted second-period payoff of the remaining $n-k$ exploiting agents, decomposed into two components: their payoff when none of their contacts in the exploring group received a high output, and their payoff when at least one did.

The following lemma examines the marginal value of adding one more explorer, $\Delta u_k\coloneq u_{k+1}-u_k$, which will prove useful when we characterize the equilibrium social surplus and the social optimum. We use the notation $Q_a(b)\coloneq\sum_{m\leq b}q_a(m)$ for the cumulative distribution of $q_a$, with the convention $Q_0(0)=q_0(0)=1$.
\begin{lemma}
\label{lem: Delta_u}
The marginal value of one more exploring agent takes the following form:
\begin{enumerate}[label=(\roman*)]
    \item \label{item: Deltau_1} On $\frac{\pi}{1-\pi} \leq \frac{\alpha}{\beta-\alpha(1-\beta)}$, or equivalently $\pi \leq \tau$, it holds that
    \begin{equation}
    \label{eq: Deltau_1}
    \begin{gathered}
        \Delta u_k(\pi)= (1-\delta)\big(\pi \beta-\alpha(1-\pi \beta)\big) \\ +\delta \pi (\beta-\alpha(1-\beta)) \beta(1-p\beta)^k \left(1+(n-1)p-\frac{kp(1-p)\beta}{1-p\beta}\right)\,.
    \end{gathered}
    \end{equation}
    \item \label{item: Deltau_2} On $\frac{\pi}{1-\pi}\geq \frac{\alpha}{(\beta -\alpha(1-\beta))(1-\beta)^{k+1}}$, it holds that $\Delta u_k(\pi) = (1-\delta)\big(\pi \beta-\alpha(1-\pi\beta)\big)$.
    \item \label{item: Deltau_3} On $\frac{\alpha}{(\beta -\alpha(1-\beta))(1-\beta)^r} \leq \frac{\pi}{1-\pi}\leq \frac{\alpha}{(\beta -\alpha(1-\beta))(1-\beta)^{r+1}}$ and $k\geq r\geq 0$, it holds that $$\Delta u_k(\pi) = (1-\delta)(\pi \beta-\alpha(1-\pi \beta)) + \delta \big[\pi (\beta-\alpha(1-\beta))B_k-\alpha(1-\pi)A_k\big]\,,$$
    where
    \begin{equation*}
    \begin{gathered}
    A_k(\pi) \coloneq (k+1) Q_k(r-1)-kQ_{k-1}(r-1) - (n-k)Q_k(r) +(n-k-1)Q_{k+1}(r)\,\\
    B_k(\pi) \coloneq -(k+1)(1-\beta)\sum_{m=r}^{k} q_k(m)(1-\beta)^m + k(1-\beta)\sum_{m=r}^{k-1} q_{k-1}(m)(1-\beta)^m \\ + (n-k)\sum_{m=r+1}^{k} q_k(m)(1-\beta)^m - (n-k-1)\sum_{m=r+1}^{k+1} q_{k+1}(m)(1-\beta)^m
    \end{gathered}
    \end{equation*}
\end{enumerate}
\end{lemma}
The proof follows directly by observing that the piecewise linear components in~\eqref{eq: social_welfare_func} are positive whenever $m+1\leq r$ and $m\leq r$ in the first and second components, respectively, and is therefore omitted.

%----------------------------------------------------------------------
\subsection{Equilibrium Social Surplus}
\label{subs: equil_soc_surplus}
This subsection studies equilibrium social surplus under pure and symmetric mixed strategies. We first examine the finite-population pure-strategy qualification and then turn to the symmetric-mixed benchmark and its large-population limit.
\subsubsection{Finite-Population Pure-Strategy Qualification}
We begin with pure strategies.
For $\pi\leq\underline{\pi}$, no agent explores and the equilibrium social surplus is zero. 
Figure~\ref{fig: finite_n_equil_ss} plots $u_{k_n,n}(\pi)/n$ for finite $n$ and a fixed $\pi\in(\underline{\pi},\tau)$. The curve spans the full-exploration regime and, after the first threshold, asymmetric pure-strategy regimes as $\lambda$ varies. Within a fixed equilibrium region (with $k_n$ constant), an increase in $\lambda$ shifts the distribution of $M$ in the sense of first-order stochastic dominance, which by the above representation \emph{increases} the equilibrium social surplus. The next proposition shows that at every threshold where the economy undergoes an equilibrium regime change, the social surplus falls---confirming the intuition from the two-player case.
\begin{figure}[htbp]
    \centering
    \pgfplotsset{normalsize}
    \begin{tikzpicture}[scale=1]
    \begin{axis}[smooth,
        axis line style={thick},
        x label style={at={(axis description cs:1,0)},anchor=south},
        y label style={at={(axis description cs:0,1)},anchor=east,rotate=-90},
        yticklabels={},
        ytick=\empty,
        xmax = 11,
        ymin = 0,
        axis lines=left,
        xlabel={$\lambda$},
        x label style={below},
        ylabel={$\frac{u_{k_n,n}}{n}$}
        ]
        \addplot[thick] table[col sep=comma,samples=10000] {csv/finite_n_equil_ss.csv};
    \end{axis}   
    \end{tikzpicture}
    \captionsetup{width=.6\linewidth}
    \caption{Finite-$n$ average of pure-strategy equilibrium social surplus $[\pi=0.5, \delta=0.045, \alpha=0.45,\beta=0.6, n=16]$}
    \label{fig: finite_n_equil_ss}
\end{figure}
\begin{proposition}[Finite-population pure-strategy regime changes]
\label{prop: equil_soc_surplus_discont}
The pure-strategy equilibrium social surplus falls discontinuously at every $\lambda$ where the economy undergoes an equilibrium regime change.
\end{proposition}
Figure~\ref{fig: finite_n_equil_ss} shows that the pure-strategy equilibrium social surplus increases in $\lambda$ within each equilibrium region, with discontinuous drops at regime-change thresholds. The largest drop occurs at the transition from full exploration to the intermediate region. This figure mirrors panel (b) of Figure~\ref{fig: SS}, with multiple drops reflecting the richer equilibrium structure for $n>2$.

\begin{remark}[Finite-population regime changes]
\label{rem: finite_integer_effect}
The drops in Figure~\ref{fig: finite_n_equil_ss} should be read as a finite-population pure-strategy phenomenon. Within a fixed regime, higher connectivity improves diffusion and raises surplus. At a regime-change threshold, however, the integer number of explorers falls by one, so the economy loses a discrete unit of knowledge production. This explains the downward jumps. The large-$n$ analysis below nets the production and diffusion channels after the explorer share has a continuous limit; in that benchmark the jumps vanish and per-capita surplus is weakly increasing in $\lambda$. The symmetric mixed-strategy equilibrium leads to the same limiting surplus.
\end{remark}

\paragraph{Large $n$ limit.} Define
\begin{equation}
\label{eq:ubar_infty_def}
\bar{u}_\infty(\pi,\lambda)\coloneq\lim_{n\to\infty}\frac{u_{k_n,n}(\pi)}{n}\,.
\end{equation}
By Proposition~\ref{prop: equil_num_explr_agents}, as $n\to\infty$, the fraction $k_n/n\to\kappa$ and we use this to show in the next proposition that 
\begin{subnumcases}{\bar{u}_\infty(\pi,\lambda)=\label{eq:ubar_inf}}
\label{eq:ubar_inf_1}0\,, & $\pi \leq \underline{\pi}$\,,\\[6pt]
\label{eq:ubar_inf_2}(1-\delta)\beta^{-1}\!\left(\pi\beta-\alpha(1-\pi\beta)\right)+\delta\pi(\beta-\alpha(1-\beta))\,, & $\underline{\pi} \leq \pi < \bar{\pi}_\infty^{\mathrm{local}}$\,,\\[12pt]
(1-\delta)\!\left(\pi\beta-\alpha(1-\pi\beta)\right)+\delta\pi(\beta-\alpha(1-\beta))\!\left(1-(1-\beta)e^{-\lambda\beta}\right)\nonumber\\
\label{eq:ubar_inf_3}\quad+\delta\,\BE_{M\sim\mathsf{Poisson}(\lambda)}\!\left[\!\left(\pi(\beta-\alpha(1-\beta))(1-\beta)^{M+1}-\alpha(1-\pi)\right)^{\!+}\right], & $\pi \geq \bar{\pi}_\infty^{\mathrm{local}}$\,.
\end{subnumcases}
\begin{proposition}[Large-population local-ER welfare and connectivity]
\label{prop: asympt_equil_soc_surplus}
$\bar{u}_\infty$ follows~\eqref{eq:ubar_inf} and is weakly increasing in $\lambda$ for every fixed $\pi$. Moreover, as a function of $\lambda$:
\begin{enumerate}[label=(\roman*)]
    \item \label{enum:ubar_1}For $\pi\in\left(\underline{\pi},\tau\right)$, there exists a threshold $\lambda(\pi)$ such that $\bar{u}_\infty(\pi,\lambda)$ is constant in $\lambda$ for $\lambda\geq\lambda(\pi)$, and follows~\eqref{eq:ubar_inf_3} for $\lambda<\lambda(\pi)$.
    \item \label{enum:ubar_2}For $\pi\geq\tau$, $\bar{u}_\infty$ follows~\eqref{eq:ubar_inf_3} for all $\lambda$.
\end{enumerate}
\end{proposition}
Figure~\ref{fig: lambda_effect_soc_surplus} illustrates the asymptotic average equilibrium social surplus against $\lambda$ for two values of $\pi$. For intermediate $\pi\in(\underline{\pi},\tau)$, the surplus becomes flat in $\lambda$ once $\lambda\geq\lambda(\pi)$ (Panel~\ref{fig: moderate_pi_lambda_effect}). The intuition is that an increase in $\lambda$ induces more free-riding and fewer explorers, reducing the social cost of first-period exploration but also shrinking the benefits from second-period information exchange; in the large-$n$ limit, these two effects exactly offset. This is also the region where, for finite $n$, the surplus exhibits discontinuous jumps at each regime change (Figure~\ref{fig: finite_n_equil_ss}); as $n\to\infty$, these jumps vanish and the per-capita surplus becomes flat in $\lambda$.
\begin{figure}[htbp]
\begin{center}
\begin{subfigure}[t]{.45\textwidth}
\begin{center}
\pgfplotsset{normalsize}
\begin{tikzpicture}[scale=0.9]
\begin{axis}[smooth,
    axis line style={thick},
    x label style={at={(axis description cs:1,0)},anchor=south},
    y label style={at={(axis description cs:0,1)},anchor=east,rotate=-90},
    yticklabels={},
    ytick=\empty,
    xtick={7.6076},
    xticklabels  = {$\lambda(\pi)$},
    ymin = 0.05,
    ymax = 0.061,
    xmax = 22,
    axis lines=left,
    xlabel={$\lambda$},
    ylabel={$\bar{u}_\infty$},
    x label style={below},
    ]
    \addplot[line width=1pt, thick] table[col sep=comma] {csv/moderate_pi_soc_surplus.csv};
    
    \draw [dashed] (axis cs: 7.6076,0.05) --(7.6076,0.0599);
\end{axis}     
\end{tikzpicture}
\caption{$\underline{\pi}< \pi <\tau$}
\label{fig: moderate_pi_lambda_effect}
\end{center}
\end{subfigure}\quad
\begin{subfigure}[t]{.45\textwidth}
\begin{center}
\pgfplotsset{normalsize}
\begin{tikzpicture}[scale=0.9]
\begin{axis}[smooth,
    axis line style={thick},
    x label style={at={(axis description cs:1,0)},anchor=south},
    y label style={at={(axis description cs:0,1)},anchor=east,rotate=-90},
    yticklabels={},
    ytick=\empty,
    xtick=\empty,
    ymin = 0.2555,
    ymax = 0.262,
    xmax = 22,
    axis lines=left,
    xlabel={$\lambda$},
    ylabel={$\bar{u}_\infty$},
    x label style={below},
    ]
    \addplot[line width=1pt, thick] table[col sep=comma] {csv/large_pi_soc_surplus.csv};
\end{axis}
\end{tikzpicture}
\caption{$\pi \geq \tau$}
\label{fig: large_pi_lambda_effect}
\end{center}
\end{subfigure}
\caption{Effect of $\lambda$ on $\bar{u}_\infty$}
\label{fig: lambda_effect_soc_surplus}
\end{center}
\end{figure}

%-------------------------------------------
\subsubsection{Mixed-Strategy Benchmark and Large-\texorpdfstring{$n$}{n} Limit}
We now study the equilibrium social surplus when agents use symmetric mixed strategies in the intermediate region, each exploring independently with probability $\mu^e_n$ as characterized in~\eqref{eq:equil_mix_prob}. Let $k^{\text{mixed}}_n\sim\mathsf{Bin}(n,\mu^e_n)$ denote the resulting number of explorers.\footnote{This differs from the pure-strategy equilibrium number of explorers $k_n$ in Proposition~\ref{prop: equil_num_explr_agents}.}

The average equilibrium social surplus under symmetric mixed strategies is
\begin{equation*}
\bar{u}_n^{\text{mixed}}\coloneq\frac{u_{k^{\text{mixed}}_n,n}}{n}\,.
\end{equation*}
For every finite $n$, Corollary~\ref{cor:mixed_expected_utility} gives its expectation. Thus, holding an interior prior fixed, $\BE[\bar{u}_n^{\text{mixed}}]$ is constant in $\lambda$ even though $\mu^e_n$ decreases with $\lambda$.
Our main result here is that $\bar{u}_n^{\text{mixed}}$ converges \emph{almost surely} to $\bar{u}_\infty$, the limiting average social surplus under pure strategies~\eqref{eq:ubar_infty_def}.
\begin{theorem}[Mixed and pure large-population welfare]
\label{thm:ubar_mixed_infty}
Fix $\lambda>0$ and $\pi\in(\underline{\pi},\bar{\pi}_\infty^{\mathrm{local}}(\lambda))$. Along the sequence of interior symmetric mixed-strategy equilibria, $\bar{u}_n^{\mathrm{mixed}}$ converges almost surely to $\bar{u}_\infty$.
\end{theorem}
The proof uses concentration of the Binomial explorer share together with the fact that $\mu^e_n\to\kappa(\pi)$ on the intermediate region (Proposition~\ref{prop:equil_mix_prob}), driving the same limiting average social surplus as in the pure-strategy case.

%--------------------------------------------------------------------------
%---------------------------------------------------------------------
\subsection{The Planner's Problem}
\label{subs:planner}
\label{subs:efficiency}
Agents do not internalize the information their experiments provide to others. We therefore characterize the planner's optimal thresholds and symmetric mixing probability and compare them with equilibrium. 
\subsubsection{The Social Optimum and the Under-Exploration Wedge}
\label{subs:optimal_thresholds}
One would expect the social optimum to mirror the two-player case: exploration (respectively, exploitation) is socially optimal when the initial belief exceeds (respectively, falls below) some threshold. Establishing this in the multi-agent economy, however, requires a more involved analysis.

To determine when adding one more explorer raises social welfare, we need to sign $\Delta u_k=u_{k+1}-u_k$. Lemma~\ref{lem: Delta_u} decomposed $\Delta u_k$ into two components $A$ and $B$. The following lemma is the cornerstone of the social optimum analysis, and its proof relies on the first-order stochastic dominance relation $\mathsf{Bin}(k+1,p)\succeq\mathsf{Bin}(k,p)$.
\begin{lemma}
\label{lem: B_A_inequality}
For every $k\geq r\geq 0$,
\begin{equation*}
    B_k \geq \max\big\{(1-\beta)^r A_k, (1-\beta)^{r+1}A_k \big\}.
\end{equation*}
\end{lemma}
Lemma~\ref{lem: B_A_inequality} gives tight control over $\Delta u_k$ on $[\tau,1]$. For $\pi\leq\tau$, we need an additional result:
\begin{lemma}
\label{lem: decreasing_marginals}
For every fixed $\pi\leq\tau$, the marginal value $\Delta u_k(\pi)$ is decreasing in $k$.
\end{lemma}
This follows immediately from part~\ref{item: Deltau_1} of Lemma~\ref{lem: Delta_u}. It falls short of establishing diminishing returns globally, confining the claim to the region where the initial belief is small. Together with Lemmas~\ref{lem: Delta_u} and~\ref{lem: B_A_inequality}, however, it suffices to characterize the regions where full exploitation and full exploration are socially optimal.
\begin{theorem}[Social optimum]
\label{thm: social_optimum}
The socially optimal outcome is full exploitation if and only if $\pi \leq \underline{\pi}^*$, and full exploration if and only if $\pi \geq \bar\pi^*$. Moreover, on $[0,\underline{\pi}^*]$ the social surplus is decreasing in $k$ ($\Delta u_k \leq0$), and on $[\bar\pi^*,1]$ it is increasing in $k$ ($\Delta u_k \geq0$). The cutoff points are:
\begin{gather*}
    \underline{\pi}^*_n= \dfrac{\alpha (1-\delta)}{\beta\left[(1+\alpha)(1-\delta) + \delta(\beta-\alpha(1-\beta))(1 + (n-1)p)\right]}\,,\\
    \bar\pi^*_n = \frac{\alpha(1-\delta)}{\beta\left[(1+\alpha)(1-\delta)+\delta(\beta-\alpha(1-\beta))\,(1-p\beta)^{n-1}\left(1+(n-1)p-\frac{(n-1)\,p(1-p)\beta}{1-p\beta}\right)\right]}\,.
\end{gather*}
\end{theorem}

Recall that $p=\lambda/n$. Taking $n\to\infty$, the limiting lower and upper cutoff points for the optimality of full exploitation and full exploration are:
\begin{equation*}
    \begin{gathered}
       \underline{\pi}^*_\infty \coloneq \lim_{n \to \infty} \underline{\pi}^*_n = \frac{\alpha(1-\delta)}{\beta\left[(1+\alpha)(1-\delta) + \delta(\beta-\alpha(1-\beta))(1 + \lambda)\right]}\,,\\
       \bar\pi^*_\infty \coloneq \lim_{n \to \infty} \bar\pi^*_n=\frac{\alpha(1-\delta)}{\beta\left[(1+\alpha)(1-\delta)+\delta(\beta-\alpha(1-\beta))\,e^{-\lambda \beta}\left(1+\lambda(1-\beta)\right)\right]}\,.
    \end{gathered}
\end{equation*}
Recall that the equilibrium exploitation threshold in the multi-agent economy coincides with that of the two-player case, namely $\underline{\pi}$ from Proposition~\ref{prop: 2_players_imperfect_comm}. Comparing the socially optimal exploitation threshold $\underline{\pi}^*_\infty$ with $\underline{\pi}$ confirms that the equilibrium exhibits over-exploitation relative to the social optimum. Comparing the socially optimal exploration threshold $\bar{\pi}^*_\infty$ with the asymptotic equilibrium exploration threshold~\eqref{eq: local_exploration_threshold_limit} similarly confirms under-exploration.

\paragraph{Effect of $\lambda$ on the optimal exploration cutoff.} The optimal exploration cutoff $\bar{\pi}^*_\infty$ is eventually increasing in $\lambda$; when $\beta<1/2$, it initially decreases before turning upward. When present, this non-monotonicity is not the same phenomenon as the finite-population jumps in Figure~\ref{fig: finite_n_equil_ss}; it is a property of the planner's limiting marginal value of exploration. To understand this marginal value, we examine the contribution of the $n$-th explorer to social surplus (i.e., $\Delta u_{n-1}$), and specifically its effect on the positive externality of community exploration for an agent whose first-period exploration failed. This is captured by the marginal change in the third line of~\eqref{eq: social_welfare_func}, namely
\begin{equation}
\label{eq: n_th_effect}
    \delta \pi(\beta-\alpha(1-\beta))(1-\beta)\left[n\sum_{m=0}^{n-1}q_{n-1}(m)\big(1-(1-\beta)^m\big)-(n-1)\sum_{m=0}^{n-2}q_{n-2}(m)\big(1-(1-\beta)^m\big)\right]\,.
\end{equation}
We employ a coupling argument to further illuminate the marginal change in~\eqref{eq: n_th_effect} and its response to $\lambda$. Suppose that in the high state, an agent who chose the risky arm failed in the first period, which occurs with probability $\pi(1-\beta)$. Let $X\sim\mathsf{Bin}(n-2,\lambda/n)$ denote the number of her second-period contacts excluding herself and the candidate $n$-th agent. Setting aside the base probability $\pi(1-\beta)$ and expected payoff $(\beta-\alpha(1-\beta))$, the bracketed difference in~\eqref{eq: n_th_effect} equals
\begin{equation*}
    n\Big(\BE_{X,Z}\left[1-(1-\beta)^{X+Z}\right]-\BE_X\left[1-(1-\beta)^X\right]\Big)+\BE_X\left[1-(1-\beta)^X\right]\,,
\end{equation*}
where $Z\sim\mathsf{Bernoulli}(\lambda/n)$ indicates whether the candidate $n$-th agent is a contact. This simplifies to
\begin{equation*}
    \lambda\beta\left(1-\lambda\beta/n\right)^{n-2}+1-\left(1-\lambda\beta/n\right)^{n-2}\to1+(\lambda\beta-1)e^{-\lambda\beta}\,.
\end{equation*}
This expression illustrates how connectivity changes the positive externality of the $n$-th agent's exploration. The exact cutoff pattern follows directly from differentiating the factor $e^{-\lambda\beta}(1+\lambda(1-\beta))$ in $\bar{\pi}^*_\infty$:
\begin{equation*}
    \frac{\d}{\d\lambda}\left[e^{-\lambda\beta}(1+\lambda(1-\beta))\right]=e^{-\lambda\beta}\left[(1-2\beta)-\lambda\beta(1-\beta)\right].
\end{equation*}
Since $\bar{\pi}^*_\infty$ moves inversely with this factor, it is increasing throughout when $\beta\geq1/2$. When $\beta<1/2$, it decreases for $\lambda<(1-2\beta)/(\beta(1-\beta))$ and increases thereafter. Figure~\ref{fig:lambda} plots the large-$n$ limits of the equilibrium and optimal exploration cutoffs as a function of $\lambda$ in the local economy.
\begin{figure}[htbp]
\centering
\begin{tikzpicture}
\tikzmath{
\Del = 0.25;
\Al = 0.5;
\Be = 0.4;
}
\begin{axis}[smooth,
  axis line style={thick},
  x label style={at={(axis description cs:1,0)},anchor=south},
  yticklabels={},
  ytick=\empty,
  xmax=9,
  axis lines=left,
  xlabel={$\lambda$},
  x label style={below},
  ymin=0.81,
]
\addplot[thick, samples=400, domain=0:8]
{
  (\Al*(1-\Del)) /
  (\Be*((1+\Al)*(1-\Del)
         + \Del*(\Be - \Al*(1-\Be)) * exp(-x*\Be)))
}
[xshift=-10pt]
node[pos=0.3, above] {$\bar{\pi}_\infty$};
\addplot[dashed, thick, samples=400, domain=0:8]
{
  (\Al*(1-\Del)) /
  (\Be*((1+\Al)*(1-\Del)
         + \Del*(\Be - \Al*(1-\Be))
           * exp(-x*\Be) * (1 + x*(1-\Be))))
}
[xshift=10pt]
node[pos=0.3, below] {$\bar{\pi}^*_\infty$};
\end{axis}
\end{tikzpicture}
\captionsetup{width=.5\linewidth}
\caption{Effect of $\lambda$ on equilibrium and optimal exploration thresholds
$[\alpha=0.5,\, \beta=0.4,\, \delta = 0.25]$}
\label{fig:lambda}
\end{figure}
%----------------------------------------

\subsubsection{Equilibrium versus Optimal Mixing Probability}
\label{subs:optimal_mixing}
Having characterized the optimal exploitation and exploration thresholds in the previous section, we now turn to the optimal exploration probability $\mu^*$---the mixing probability prescribed by the planner---and compare it to the equilibrium outcome $\mu^e$.

A benevolent planner who randomizes symmetrically across players with probability $\mu$ induces $k\sim\mathsf{Bin}(n,\mu)$ explorers, yielding the expected social surplus:
\begin{equation*}
S(\pi;\mu)\coloneq\BE_{k \sim\mathsf{Bin}(n,\mu)}\left[u_{k,n}(\pi)\right]\,.
\end{equation*}
Since the expected first-period payoff of exploration is positive for $\pi>\tau$, both the equilibrium and optimal strategy prescribe full exploration there. We therefore restrict the welfare analysis under mixed strategies to $[0,\tau]$. By Lemma~\ref{lem: decreasing_marginals}, the mapping $k\mapsto u_{k,n}(\pi)$ is concave on this region. The following lemma shows that this concavity carries over to $\mu\mapsto S(\pi;\mu)$ and provides a complete characterization of an interior optimum $\mu^*\in(0,1)$; the proof follows by differentiating $S(\pi;\mu)$ with respect to $\mu$ and applying a classical result on Bernstein polynomials \citep{lorentz2012bernstein}, and is therefore omitted.
\begin{lemma}[\cite{lorentz2012bernstein}]
\label{lem:surplus_conc}
For every positive prior $\pi\leq\tau$, the mapping $\mu\mapsto S(\pi;\mu)$ is strictly concave, so the planner's problem $\max\limits_{\mu\in[0,1]}S(\pi;\mu)$ has a unique maximizer. Thus a necessary and sufficient condition for an interior optimum $\mu^*\in(0,1)$ is the first-order condition:
\begin{equation}
\label{eq:foc}
\BE_{k\sim\mathsf{Bin}(n-1,\mu^*)}\left[\Delta u_{k}(\pi)\right]=0\,.
\end{equation}
\end{lemma}
\begin{theorem}[Equilibrium vs. optimum]
\label{thm:optimal_mu}
Assume the local Erd\H{o}s-R\'enyi connections, and $\pi \leq \tau$. Let $q\coloneq 1-p\beta$. Then, the equilibrium $\mu^e$ and optimal $\mu^*$ satisfy, respectively,
\begin{subequations}
\begin{align}
\label{eq:equil}
&C_1(\pi)=C_2(\pi)\;\BE_{k\sim\mathsf{Bin}(n-1,\mu^e)}\!\left[q^k\right]\,,\\
\label{eq:optimal}
&C_1(\pi)=C_2(\pi)\;\BE_{k\sim\mathsf{Bin}(n-1,\mu^*)}\!\left[q^k\!\left(1+(n-1)p-\frac{kp(1-p)\beta}{1-p\beta}\right)\right]\,.
\end{align}
\end{subequations}
\end{theorem}
\begin{proof}
The identity for the equilibrium probability $\mu^e$ in~\eqref{eq:equil} is already shown in Proposition~\ref{prop:equil_mix_prob}. The condition for $\mu^*$ in~\eqref{eq:optimal} follows by substituting $\Delta u_k$ from~\eqref{eq: Deltau_1} into the first-order condition~\eqref{eq:foc}.
\end{proof}
The term $\left(1+(n-1)p-\frac{kp(1-p)\beta}{1-p\beta}\right)$ in~\eqref{eq:optimal} captures the positive externality of an additional explorer. For $\pi\leq\tau$, the equilibrium condition~\eqref{eq:equil} reflects that agents explore only when the informational gain ($C_2$) outweighs the first-period loss ($C_1$), but they ignore the information they provide to others---giving rise to the wedge between $\mu^e$ and $\mu^*$. The following corollary shows that the equilibrium mixing probability is always below the social optimum, confirming that under-exploration is robust to mixed strategies. 
\begin{corollary}
\label{cor:equil_vs_optimal_mu}
For every positive prior $\pi<\tau$, $\mu^*\geq\mu^e$.
\end{corollary}
\begin{proof}
Define $\Gamma(k)\coloneq 1+(n-1)p-\frac{kp(1-p)\beta}{1-p\beta}\geq 1$, and let
\begin{equation*}
    G(\mu)\coloneq\BE_{k\sim\mathsf{Bin}(n-1,\mu)}\!\left[q^k\right]\quad\text{and}\quad H(\mu)\coloneq\BE_{k\sim\mathsf{Bin}(n-1,\mu)}\!\left[q^k\Gamma(k)\right].
\end{equation*}
Because $q^k$ and $q^k\Gamma(k)$ decrease in $k$, $G$ and $H$ decrease in $\mu$, with $G\leq H$. If $\mu^e=0$, the result is immediate. If $\mu^e=1$, the equilibrium endpoint condition gives $C_1(\pi)/C_2(\pi)\leq G(1)\leq H(1)$, so concavity implies $\mu^*=1$. Otherwise, $C_1(\pi)/C_2(\pi)=G(\mu^e)$. If $\mu^*=1$, the result is again immediate; if not, $H(0)\geq G(0)>G(\mu^e)$ rules out $\mu^*=0$, and~\eqref{eq:optimal} gives $C_1(\pi)/C_2(\pi)=H(\mu^*)\geq G(\mu^*)$. Monotonicity of $G$ now yields $\mu^*\geq\mu^e$.
\end{proof}

%%%%%%%%%%%%%%%%%%%%%%%%%%%%%%%%%%%%%%%%%%%%%%%%%%%%%%

\section{Conclusion and Additional Discussion}
\label{sec:conclusion}

The tension between information diffusion and production in organizations and societies represents a complex and evolving challenge. On one hand, greater connectivity has allowed for unprecedented access to knowledge and ideas, fostering informed decision-making and quicker adoption of innovation. At the same time, this proliferation of information can undermine knowledge production. In a better connected organization or society, individuals are tempted to rely on knowledge shared through the network instead of experimenting with new ideas themselves. Greater connectivity can therefore lower knowledge production, while equilibrium experimentation remains below the socially optimal level.

Our analysis begins with a two-player economy, where equilibrium reveals three distinct regions based on initial beliefs, each reflecting different degrees of exploitation and exploration. Free-riding leads to over-exploitation and under-exploration relative to the social optimum. In this finite economy, pure-strategy equilibrium social surplus can also be non-monotone in the connection probability because the number of explorers changes discretely. The large-population local Erd\H{o}s-R\'enyi analysis clarifies the scope of this welfare statement: the finite-player jumps disappear, and per-capita surplus is weakly increasing in connectivity and eventually flat for intermediate beliefs. Our model is set in two periods; we expect the non-monotonicity of equilibrium social surplus to persist in longer horizons—as studied in continuous time by~\cite{board2024experimentation}—and that as agents become more patient ($\delta \to 1$), the exploration threshold decreases, since it becomes preferable to bear the exploration cost sooner rather than later.

We then analyze multi-agent economies and explore different network structures. With local connections each agent only observes the experimentation outcomes of
her immediate neighbors, whereas with global connections each agent's observable circle includes the entire set of agents who are (directly or indirectly) connected to her. In both structures, the tension between information sharing and private exploration remains a significant factor. Higher connectivity exacerbates free-riding and reduces experimentation. In the finite-$n$ local Erd\H{o}s-R\'enyi economy, pure-strategy equilibrium social surplus can fall discontinuously when a small increase in connectivity induces one fewer explorer. This discrete drop is a finite-population qualification rather than the welfare benchmark: under symmetric mixing, expected per-capita equilibrium social surplus is constant in connectivity throughout the interior mixing region. As $n\to\infty$, the pure- and symmetric mixed-strategy equilibria converge to the same per-capita surplus; this common limit is weakly increasing in connectivity and eventually flat for intermediate beliefs. Under global connections, the Erd\H{o}s-R\'enyi phase transition sharply changes exploration incentives around $\lambda=1$.

Although we showed that the tension between information diffusion and production is prevalent in different settings, interesting variations remain to be studied. For example, what happens when agents have different preferences or endowments? What happens when societies and organizations have network structures with differently connected agents? How do the production and welfare effects change in longer-horizon environments? We hope to address these issues in future work.

%%%%%%%%%%%%%%%%%%%%%%%%%%%%%%%%%%%%%%%%%%%%%%%%%%%%%%

\newpage
\appendix
\addtocontents{toc}{\protect\setcounter{tocdepth}{1}}
\section{Proof of Results in Section~\ref{sec: many_players_Econ}}
\subsection{Proof of Theorem~\ref{thm: exploration_equil}}
Suppose all but one individual explore in period one. Let $L$ be the number of successes ($y=1$) the focal agent observes; clearly $L\le M$ (her period-two contacts). Define $\pi_{\ell,m}\coloneq\BP\left(\theta=1\mid L=\ell, M=m\right)$. Then $\pi_{\ell,m}=1$ for $\ell\ge1$, and if $\ell=0$:
\begin{equation*}
    \frac{\pi_{0,m}}{1-\pi_{0,m}} = \frac{\pi}{1-\pi}(1-\beta)^m\,.
\end{equation*}
In period two, the focal agent pulls the risky arm iff $\pi_{\ell,m}>\tau$, yielding
\begin{equation*}
\left[\pi_{\ell,m}\beta-\alpha\bigl(1-\pi_{\ell,m}\beta\bigr)\right]^+.
\end{equation*}
Equivalently,
\begin{equation*}
\left[\pi_{\ell,m}\beta-\alpha(1-\pi_{\ell,m}\beta)\right]^+=\BE\left[\theta\beta-\alpha(1-\theta\beta)\mid L=\ell, M=m\right]^+\,.
\end{equation*}
If all others explore in period one, the focal agent's exploitation payoff is
\begin{equation}
\label{eq: all_w}
\begin{aligned}
     w_{n-1}(\pi) &= \delta \sum_{m=0}^{n-1}\sum_{\ell=0}^m \BP\left(L=\ell,M=m\right)\BE\left[\theta\beta-\alpha(1-\theta\beta) \big| L=\ell,M=m\right]^+\\
     &= \delta \sum_{m=0}^{n-1}\sum_{\ell=0}^m\BE\left[\theta\beta-\alpha(1-\theta\beta);\, L=\ell, M=m\right]^+\,.
\end{aligned}
\end{equation}
Let $q(m)\coloneq\BP(M=m)$ be the probability that a randomly selected agent observes the exploration outcomes of $m$ others. Then
\begin{equation*}
    \begin{gathered}
        w_{n-1}(\pi) = \delta \sum_{m=0}^{n-1}\sum_{\ell=0}^m q(m)
        \left[\pi (\beta-\alpha(1-\beta)) \binom{m}{\ell} \beta^\ell (1-\beta)^{m-\ell}-\alpha(1-\pi)1_{\{\ell=0\}}\right]^+\\
        = \delta \sum_{m=0}^{n-1} \underbrace{q(m)\left[\pi (\beta-\alpha(1-\beta)) (1-\beta)^m-\alpha(1-\pi)\right]^+}_{\BE\left[\theta\beta-\alpha(1-\theta\beta);\, M=m, L=0\right]^+}
        + \underbrace{\delta \pi(\beta-\alpha(1-\beta)) \sum_{m=0}^{n-1}q(m)\left(1-(1-\beta)^m\right)}_{\delta \BE\left[\theta\beta-\alpha(1-\theta\beta);\, L>0\right]}\,.
    \end{gathered}
\end{equation*}
Thus $w_{n-1}(\pi)$ splits into the payoff when $L=0$ and when $L>0$.

If the agent explores in period one, let $y_0\in\{-\alpha,1\}$ denote her realization of the risky arm. Her expected payoff is
\begin{equation}
\label{eq: all_v}
    \begin{gathered}
        v_n(\pi) = (1-\delta)\left(\pi \beta-\alpha(1-\pi \beta)\right)\\
        +\delta \sum_{m=0}^{n-1}\sum_{\ell=0}^m\sum_{y\in\{-\alpha,1\}} \BP\left(M=m,L=\ell,y_0=y\right)\BE\left[\theta\beta-\alpha(1-\theta\beta)\mid M=m,L=\ell,y_0=y\right]^+\,.
    \end{gathered}
\end{equation}
The discounted period-2 term splits by the period-1 realization $y_0$:
\begin{equation*}
    \begin{aligned}
        \text{discounted expected payoff} 
        &=\delta \sum_{m,\ell}\BE\left[\theta\beta-\alpha(1-\theta\beta);\, M=m, L=\ell, y_0=1\right]^+\\
        &\quad+\delta \sum_{m,\ell}\BE\left[\theta\beta-\alpha(1-\theta\beta);\, M=m, L=\ell, y_0=-\alpha\right]^+\,.
    \end{aligned}
\end{equation*}
The first (second) term conditions on $y_0=1$ ($y_0=-\alpha$). Further rearrangements imply that the sum above is equal to
\begin{equation*}
    \begin{gathered}
        \delta \pi (\beta-\alpha(1-\beta))\beta
        + \delta \sum_{m,\ell} q(m)
        \left[\pi (\beta-\alpha(1-\beta)) \binom{m}{\ell}\beta^\ell (1-\beta)^{m+1-\ell}-\alpha(1-\pi)1_{\{\ell=0\}}\right]^+\\
        = \underbrace{\delta \pi (\beta-\alpha(1-\beta))\beta}_{\delta \BE\left[\theta\beta-\alpha(1-\theta\beta);\, y_0=1\right]}
        + \delta \sum_{m=0}^{n-1} \underbrace{q(m)\left[\pi (\beta-\alpha(1-\beta))(1-\beta)^{m+1}-\alpha(1-\pi)\right]^+}_{\BE\left[\theta\beta-\alpha(1-\theta\beta);\, M=m, L=0, y_0=-\alpha\right]^+}\\
        \quad+ \underbrace{\delta \pi (\beta-\alpha(1-\beta))(1-\beta)\sum_{m=0}^{n-1}q(m)\left(1-(1-\beta)^m\right)}_{\delta \BE\left[\theta\beta-\alpha(1-\theta\beta);\, L>0, y_0=-\alpha\right]}\,.
    \end{gathered}
\end{equation*}

Formally, exploration occurs iff $v_n(\pi)>w_{n-1}(\pi)$, i.e., when the discounted gains from conclusive signals ($L>0$ or $y_0=1$) outweigh the opportunity cost when signals are inconclusive ($L=0$ and $y_0=-\alpha$):
\begin{small}
\begin{equation*}
    \begin{gathered}
        \overbrace{(1-\delta)\left(\pi\beta -\alpha(1-\pi \beta)\right)}^{\text{present payoff}}\\
        +\delta \Big(\BE\left[\theta \beta-\alpha(1-\theta \beta);y_0=1\right]
        +\BE\left[\theta \beta-\alpha(1-\theta \beta);L>0,y_0=-\alpha\right]
        -\BE\left[\theta \beta-\alpha(1-\theta \beta);L>0\right]\Big)\\
        > \delta\sum_{m=0}^{n-1} \Big(\BE\left[\theta\beta-\alpha(1-\theta \beta);M=m,L=0\right]^+-\BE\left[\theta \beta-\alpha(1-\theta \beta);M=m,L=0,y_0=-\alpha\right]^+\Big)\\
        =\text{discounted opportunity cost of exploration absent conclusive signals}.
    \end{gathered}
\end{equation*}
\end{small}
If $\pi>\tau$, the agent explores in period one: the expected first-period payoff is positive and exploration yields informational benefits, so $v_n(\pi)>w_{n-1}(\pi)$. Hence the cutoff satisfies $\bar\pi\le\tau$, and we restrict attention to $\pi\le\tau$. On $\pi\le\tau$, all $[\cdot]^+$ terms equal zero because, for any $k\ge 0$,
\begin{equation*}
    \pi\left(\beta-\alpha(1-\beta)\right)(1-\beta)^k \leq \pi\left(\beta-\alpha(1-\beta)\right) \leq \alpha(1-\pi)\,,
\end{equation*}
where the last inequality uses $\pi\le\tau$. Thus $v_n(\pi)>w_{n-1}(\pi)$ reduces to
\begin{equation*}
\begin{gathered}
    (1-\delta)\left(\pi \beta-\alpha(1-\pi \beta)\right)+\delta \pi \beta (\beta -\alpha(1-\beta))\BE\left[(1-\beta)^M\right] > 0\,.
\end{gathered}
\end{equation*}
This is equivalent to $\pi>\bar\pi$, proving~\eqref{eq: n_player_pi_bar}.\qed

%----------------------------------------------------------------------

\subsection{Proof of Lemma~\ref{lem: no_exploiting_large_pi}}
Let $q_k(m)\coloneq\BP_k(M=m)$ be the probability of observing $m$ exploration outcomes among the $k$ first-period explorers. Following the recipe of~\eqref{eq: all_w} and~\eqref{eq: all_v}, the payoff functions are
\begin{equation}
\label{eq: general_payoff_functions}
    \begin{gathered}
       w_k(\pi) = \delta \BE_k\left[ \left(\pi (\beta-\alpha(1-\beta))(1-\beta)^M-\alpha(1-\pi)\right)^+\right]+\delta \pi (\beta-\alpha(1-\beta))\BE_k\left[1-(1-\beta)^M\right]\,,\\
       v_k(\pi) = (1-\delta)\big(\pi \beta-\alpha(1-\pi \beta)\big)+\delta \pi (\beta-\alpha(1-\beta)) \beta \\
       +\delta \BE_{k-1}\left[\left(\pi (\beta-\alpha(1-\beta))(1-\beta)^{M+1}-\alpha(1-\pi)\right)^+\right]\\
       +\delta \pi (\beta-\alpha(1-\beta))(1-\beta)\BE_{k-1}\left[1-(1-\beta)^M\right]\,.
    \end{gathered}
\end{equation}
Throughout, the distribution of $M$ is the one indexed by the outer expectation: $\BE_k$ means $\BP(M=m)=q_k(m)$.

Consider the difference
\begin{equation*}
    \begin{gathered}
       v_{k+1}(\pi)-w_k(\pi) = (1-\delta)\big(\pi\beta-\alpha(1-\pi \beta)\big)+\delta \pi \left(\beta -\alpha (1-\beta)\right) \beta \BE_k\left[(1-\beta)^M\right]\\
       +\delta \BE_k\left[\big(\pi \left(\beta -\alpha (1-\beta)\right)(1-\beta)^{M+1}-\alpha(1-\pi)\big)^+\right]-\delta \BE_k\left[\big(\pi \left(\beta -\alpha (1-\beta)\right)(1-\beta)^M-\alpha(1-\pi)\big)^+\right].
    \end{gathered}
\end{equation*}
For each $m$, define
\[
J_m\coloneq\pi\left(\beta-\alpha(1-\beta)\right)(1-\beta)^m-\alpha(1-\pi).
\]
Since $J_{m+1}\leq J_m$ and the positive-part map is 1-Lipschitz,
\[
J_{m+1}^+-J_m^+\geq J_{m+1}-J_m
=-\pi\left(\beta-\alpha(1-\beta)\right)\beta(1-\beta)^m.
\]
Applying this inequality pointwise at $m=M$ in the preceding expression gives
\[
v_{k+1}(\pi)-w_k(\pi)
\geq(1-\delta)\big(\pi\beta-\alpha(1-\pi\beta)\big)>0,
\]
where the final inequality follows from $\pi>\tau$ and $\delta<1$.\qed
%----------------------------------------------------------------------
\subsection{Proof of Theorem~\ref{thm: symmetric_mixed_equil}}
To prove the theorem, we first demonstrate that on the intermediate region there exists a unique $\mu$ satisfying $v(\pi;\mu)=w(\pi;\mu)$. For existence, observe that on this region $v(\pi;0) = v_1(\pi)> w_0(\pi) = w(\pi;0)$ and $v(\pi;1)=v_n(\pi)\leq w_{n-1}(\pi)=w(\pi;1)$. Therefore, there exists at least one $\mu^e \in (0,1]$ satisfying $v(\pi; \mu^e) = w(\pi; \mu^e)$.

For uniqueness, observe that the $[\cdot]^+$ terms in equation~\eqref{eq: general_payoff_functions} vanish on the intermediate region, yielding:
\begin{equation*}
    \begin{gathered}
       w(\pi;\mu)=\delta \pi (\beta -\alpha(1-\beta)) \sum_{k=0}^{n-1}\binom{n-1}{k}\mu^k (1-\mu)^{n-1-k}\BE_k\left[1-(1-\beta)^M\right]\,,\\
       v(\pi;\mu) = (1-\delta)\left(\pi \beta-\alpha(1-\pi \beta)\right)+\delta \pi(\beta -\alpha (1-\beta)) \beta\\
      +\delta \pi (\beta-\alpha(1-\beta)) (1-\beta) \sum_{k=0}^{n-1} \binom{n-1}{k}\mu^k(1-\mu)^{n-1-k}\BE_k\left[1-(1-\beta)^M\right]\,.
    \end{gathered}
\end{equation*}
Hence $v(\pi;\mu)=w(\pi;\mu)$ is equivalent to
\begin{equation}
\label{eq: mu_cond}
\begin{gathered}
    \frac{(1-\delta)\left(\alpha(1-\pi \beta) -\pi \beta\right)}{\delta \pi  (\beta -\alpha (1-\beta)) \beta} = \sum_{k=0}^{n-1}\binom{n-1}{k}\mu^k (1-\mu)^{n-1-k}\, \BE_k\left[(1-\beta)^M\right]\\
    = \BE_{k \sim \mathsf{Bin}(n-1,\mu)}\left[\BE_k\left[(1-\beta)^M\right]\right]\,.
\end{gathered}
\end{equation}
Since the Binomial distribution $\mathsf{Bin}(n-1,\mu)$ is increasing in the sense of FOSD with respect to $\mu$, and the mapping $k \mapsto \BE_k\left[(1-\beta)^M\right]$ is decreasing in $k$ (as can be shown by a simple coupling argument), the RHS of~\eqref{eq: mu_cond} is decreasing in $\mu$. This establishes the uniqueness of $\mu^e$. 

In addition, the LHS of~\eqref{eq: mu_cond} is decreasing in $\pi$. Given that its RHS is decreasing in $\mu$, the equilibrium mixing probability $\mu^e$ must be increasing in $\pi$.

Lastly, observe that for a fixed $\mu$, one has $\mathsf{Bin}(n,\mu) \succeq_1 \mathsf{Bin}(n-1,\mu)$. Because the mapping $k \mapsto \BE_k\left[(1-\beta)^M\right]$ is decreasing in $k$, the RHS of~\eqref{eq: mu_cond} is decreasing in both $\mu$ and $n$. Since the LHS of~\eqref{eq: mu_cond} is unaffected by $n$, $\mu^e$ must be decreasing in $n$.\qed
%----------------------------------------------------------------
\section{Proof of Results in Section~\ref{sec: limit}}
\subsection{Proof of Lemma~\ref{lem: monotonicity_local_explr_threshold}}
In the local regime,
\[
\BE^{(n)}_{n-1}\left[(1-\beta)^M\right]=\left(1-\frac{\lambda\beta}{n}\right)^{n-1}\longrightarrow e^{-\lambda\beta}.
\]
Substituting this limit into~\eqref{eq: n_player_pi_bar} gives $\bar{\pi}_n\to\bar{\pi}_\infty^{\mathrm{local}}$.\qed

%---------------------------------------------------
\subsection{Proof of Lemma~\ref{lem: ordering}}
Since $p_n=\frac{\lambda}{n}\geq \frac{\lambda}{n+1}=p_{n+1}$, a coupling argument shows that on the same probability space $M_k^{(n)} \geq M_k^{(n+1)}$ almost surely, and therefore $\BE\left[(1-\beta)^{M^{(n)}_k}\right]\leq \BE\left[(1-\beta)^{M^{(n+1)}_{k}}\right]$. This in turn implies $\pi_{k,n+1} \leq \pi_{k,n}$.

Next, observe that with local connections,
\begin{equation*}
\begin{gathered}
   \BE^{(n+1)}_{k+1}\left[(1-\beta)^M\right] = \left(1-\frac{\lambda \beta}{n+1}\right)^{k+1} \leq\left(1-\frac{\lambda \beta}{n+1}\right)\left(1+\frac{\lambda \beta}{n+1}\right)^{-k}\\
   \leq\left(1-\frac{\lambda \beta}{n+1}\right)\left(1+\frac{k\lambda \beta}{n+1}\right)^{-1}.
\end{gathered}
\end{equation*}
Moreover,
\begin{equation*}
   \BE^{(n)}_{k}\left[(1-\beta)^M\right] =\left(1-\frac{\lambda \beta}{n}\right)^{k} \geq 1-\frac{k \lambda \beta}{n}\,.
\end{equation*}
For large $n$, one can readily verify that
\begin{equation*}
    \left(1-\frac{\lambda \beta}{n+1}\right) \leq \left(1-\frac{k \lambda \beta}{n}\right)\left(1+\frac{k\lambda \beta}{n+1}\right)\,.
\end{equation*}
Therefore,
\begin{equation*}
    \BE^{(n+1)}_{k+1}\left[(1-\beta)^M\right] \leq \BE^{(n)}_{k}\left[(1-\beta)^M\right]\,,
\end{equation*}
which implies $\pi_{k+1,n+1}\geq \pi_{k,n}$ for sufficiently large $n$. A similar argument shows $\pi_{k,n+1} \geq \pi_{k-1,n}$, concluding the proof.\qed
%-----------------------------------------------------------------------
\subsection{Proof of Proposition~\ref{prop: equil_num_explr_agents}}
For every $\pi \leq \underline{\pi}$, the full exploitation equilibrium prevails, thus $k_n(\pi)=0$. For every $\pi>\bar{\pi}_\infty^{\mathrm{local}}$, Lemma~\ref{lem: monotonicity_local_explr_threshold} implies that $\pi>\bar{\pi}_n$ for all sufficiently large $n$, hence $k_n(\pi)=n$. At $\pi=\bar{\pi}_\infty^{\mathrm{local}}$, the cutoff characterization gives
\[
k_n-1<\frac{-\lambda\beta}{\log(1-\lambda\beta/n)}\leq k_n.
\]
Since $n\log(1-\lambda\beta/n)\to-\lambda\beta$, it follows that $k_n/n\to1$. Therefore, it remains to examine the limiting behavior of $k_n(\pi)/n$ on the intermediate region $(\underline{\pi},\bar{\pi}_\infty^{\mathrm{local}})$, where asymmetric equilibria prevail. By Theorem~\ref{thm: assymetric_equil}, there will be $k_n$ explorers in the equilibrium if and only if $\pi_{k_n-1, n} < \pi \leq \pi_{k_n, n}$. According to equation~\eqref{eq: asymmetric_equil_cond}, this condition translates to
\begin{equation*}
\begin{gathered}
   \BE_{k_n}^{(n)}\left[(1-\beta)^M\right] \leq \frac{(1-\delta)\left(\alpha(1-\pi\beta) -\pi \beta \right)}{\delta \pi\beta (\beta -\alpha(1-\beta))} < \BE_{k_n-1}^{(n)}\left[(1-\beta)^M\right] \\
    \Leftrightarrow (k_n-1) < \frac{\log \left(\frac{(1-\delta)\left(\alpha(1-\pi\beta) -\pi \beta \right)}{\delta \pi\beta (\beta -\alpha(1-\beta))}\right)}{\log\left(1-\lambda \beta/n\right)}\leq k_n\,.
\end{gathered}
\end{equation*}
Therefore,
\begin{equation*}
    \lim_{n \to \infty} \frac{k_n}{n} = \frac{\log \left(\frac{(1-\delta)\left(\alpha(1-\pi\beta) -\pi \beta \right)}{\delta \pi\beta (\beta -\alpha(1-\beta))}\right)}{\lim_{n \to \infty}n \log\left(1-\lambda \beta/n\right)}=\frac{1}{\lambda \beta}\log\frac{\delta \pi\beta (\beta -\alpha(1-\beta))}{(1-\delta)\left(\alpha(1-\pi\beta)-\pi\beta\right)} = \frac{1}{\lambda \beta} \log \frac{C_2(\pi)}{C_1(\pi)}\,.
\end{equation*}\qed
%-------------------------------------------------------------------

\subsection{Proof of Proposition~\ref{prop:equil_mix_prob} and Corollary~\ref{cor:mixed_expected_utility}}
\label{app:proof_mixed_utility}
From~\eqref{eq: mu_cond},
\begin{equation*}
    C_1(\pi) = C_2(\pi)\;\BE_{k\sim\mathsf{Bin}(n-1,\mu)}\!\left[\BE_k\left[(1-\beta)^M\right]\right].
\end{equation*}
In the local ER graph with $k$ explorers, the number of contacts $M\sim\mathsf{Bin}(k,p)$ where $p=\lambda/n$, so the inner expectation equals $(1-\beta p)^k$. Therefore,
\begin{equation*}
    C_1(\pi) = C_2(\pi)\;\BE_{k\sim\mathsf{Bin}(n-1,\mu)}\!\left[(1-\beta p)^k\right] = C_2(\pi)(1-\mu\beta p)^{n-1}.
\end{equation*}
Taking logarithms of both sides yields~\eqref{eq:equil_mix_prob}. The displayed form makes the comparative static immediate: holding an interior prior fixed, the bracketed term is positive and independent of $\lambda$, so $\mu_n^e$ is strictly decreasing in $\lambda$ while the prior remains in the interior mixing region.

On that region, the payoff from exploitation against mixing probability $\mu$ is
\begin{equation*}
    w(\pi;\mu)=\delta\pi\bigl(\beta-\alpha(1-\beta)\bigr)
    \left[1-(1-\mu\beta p)^{n-1}\right].
\end{equation*}
At $\mu=\mu^e_n$, the equilibrium condition gives $(1-\mu^e_n\beta p)^{n-1}=C_1(\pi)/C_2(\pi)$. Substituting this identity into $w$ gives
\[
    w(\pi;\mu_n^e)=\delta\pi\bigl(\beta-\alpha(1-\beta)\bigr)
    \left(1-\frac{C_1(\pi)}{C_2(\pi)}\right).
\]
By indifference, $v(\pi;\mu_n^e)=w(\pi;\mu_n^e)$. Ex ante symmetry then implies that the expected per-capita social surplus equals this common payoff, proving Corollary~\ref{cor:mixed_expected_utility}.

To find the limiting $\mu^e_n$, set $R(\pi)\coloneq C_1(\pi)/C_2(\pi)\in\mathbb{R}_+$. As $n\to\infty$,
\begin{equation*}
R(\pi)^{1/(n-1)}=\ee^{\frac{1}{n-1}\log R(\pi)}\approx 1+\frac{\log R(\pi)}{n-1}\,,
\end{equation*}
and therefore
\begin{equation*}
    \lim_{n\to\infty}\mu^e_n=\lim_{n\to\infty}\frac{n}{\lambda\beta}\left(1-1-\frac{1}{n-1}\log\frac{C_1(\pi)}{C_2(\pi)}\right)=\frac{1}{\lambda\beta}\log\frac{C_2(\pi)}{C_1(\pi)}\,,
\end{equation*}
which justifies~\eqref{eq:limit_equil_mix_prob}.\qed

%------------------------------------------------------------------
\subsection{Proof of Proposition~\ref{prop: asymptotic_pi_bar}}
We only present the proof in the subcritical regime ($\lambda \leq 1$ and thus $\zeta(\lambda)=1$), showing that the asymptotic distribution of the size of a typical connected component is $B_1(\lambda) \stackrel{d}{=}\mathsf{Borel}(\lambda)$. The convergence of $|\mathcal{C}|$ to the mixture Borel distribution in the supercritical regime follows by simple conditioning on the extinction event. Henceforth, we assume $\lambda \leq 1$.

First, we show how the size of the connected component $|\mathcal{C}|$ in an ER random graph with parameters $(n,p=\lambda/n)$ can be approximated by the total number of descendants in a branching process with $\mathsf{Bin}(n,p)$ offspring distribution, which we denote by $B$. We use $\BP_{n,p}$ to refer to the distribution of $B$. Theorems~4.2 and~4.3 of \cite{van2016random} jointly state that:
\begin{equation*}
    \BP_{n-k,p}\left(B \geq k\right) \leq \BP\left(|\mathcal{C}| \geq k\right) \leq \BP_{n,p}\left(B \geq k\right)\,.
\end{equation*}
Next, we examine how the total number of progenies in a Binomial branching process with parameters $(n,p)$ can be approximated by a branching process with $\mathsf{Poisson}(np)$ offspring distribution, which we denote by $T$. We use $\BP_\lambda$ to denote the distribution of the branching process with $\mathsf{Poisson}(\lambda)$ offspring distribution. Let $\lambda = np$ and fix $k \in \BN$. Then, Theorem~3.20 in \cite{van2016random} implies that
\begin{equation*}
    \big|\BP_{n,p}(B \geq k) - \BP_\lambda(T \geq k)\big|\leq  \frac{\lambda^2 k}{n}\,.
\end{equation*}
Combining the last two relations yields
\begin{equation*}
    \BP_{\lambda(1-k n^{-1})} (T\geq k) -\frac{\lambda^2(n-k) k}{n^2} \leq \BP\left(|\mathcal{C}| \geq k\right) \leq \BP_\lambda(T\geq k)+\frac{\lambda^2 k}{n}\,.
\end{equation*}
For a fixed $k \in \BN$, let $\lambda_n \coloneq  \lambda(1-k n^{-1})$. Then, using the method of characteristic functions, one can show $\BP_{\lambda_n}$ weakly converges to $\BP_\lambda$ as $n\to \infty$ (Theorem~6.3 in~\cite{kallenberg2021foundations}). This in turn implies $\BP_{\lambda_n}(T<k) \to \BP_{\lambda}(T<k)$, and hence $\BP_{\lambda_n}(T\geq k) \to \BP_{\lambda}(T\geq k)$. Combining this with the above inequality, we conclude that for every $k \in \BN$,
\begin{equation}
\label{eq: tail_conv}
    \lim_{n \to \infty}\BP(|\mathcal{C}| \geq k) = \BP_\lambda(T\geq k)\,.
\end{equation}
Let $\bar{\BN} = \BN \cup \{\infty\}$; then $|\mathcal{C}|$ and $T$ are $\bar{\BN}$-valued random variables, with $\bar{\BN}$ equipped with the one-point compactification topology. The limiting result in~\eqref{eq: tail_conv} therefore implies weak convergence of $|\mathcal{C}|$ to $T$. The distribution of the total progeny in a Poisson branching process is known to follow the Borel distribution $\mathsf{Borel}_1(\lambda)$ (see Theorem~3.16 of \cite{van2016random}), concluding the proof of part~\ref{item: weak_conv}. 

A similar argument combined with conditioning on the extinction event in the supercritical regime implies that $|\mathcal{C}|$ converges in distribution to $\mathsf{Borel}_{\zeta(\lambda)}(\lambda)$ for $\lambda > 1$. Given our definition of $\zeta(\lambda)\equiv 1$ for $\lambda \leq 1$, we obtain convergence across all $\lambda \geq 0$, proving part~\ref{item: weak_conv}.

Part~\ref{item: asymp_pi_bar} follows immediately: $x \mapsto (1-\beta)^{x-1}$, extended by zero at $\infty$, is bounded and continuous, so by the weak convergence established above,
\begin{equation*}
    \lim_{n \to \infty}\BE\left[(1-\beta)^{|\mathcal{C}|-1}\right] = \BE\left[(1-\beta)^{B_{\zeta(\lambda)} (\lambda)-1}\right]\,,
\end{equation*}
justifying equation~\eqref{eq: asymp_pi_bar}.\qed
%----------------------------------------------------------------------------
\subsection{Proof of Lemma~\ref{lem: mixture_borel_PGF}}
We consider the subcritical and supercritical regimes separately.

\medskip
\noindent\textbf{Subcritical regime ($\lambda \leq 1$).}
In this case $\zeta(\lambda) = 1$, so $B_{\zeta(\lambda)}(\lambda) = B(\lambda)$ and $\psi_\lambda(z) = \BE\left[z^{B(\lambda)}\right]$. Since the total progeny of a branching process satisfies the distributional recursion $B(\lambda) = 1 + \sum_{i=1}^{N} B_i(\lambda)$, where $N \sim \mathsf{Poisson}(\lambda)$ and the $B_i(\lambda)$ are i.i.d.\ copies of $B(\lambda)$ independent of $N$, we have
\begin{align*}
    \psi_\lambda(z) &= \BE\left[z^{1 + \sum_{i=1}^{N} B_i(\lambda)}\right] = z \, \BE\left[\prod_{i=1}^{N} z^{B_i(\lambda)}\right] = z \, \BE\left[\psi_\lambda(z)^N\right] = z \, \mathrm{e}^{\lambda(\psi_\lambda(z)-1)},
\end{align*}
where the last step uses the probability-generating function of the Poisson distribution.

\medskip
\noindent\textbf{Supercritical regime ($\lambda > 1$).}
By definition of the mixture Borel distribution, with probability $1-\zeta(\lambda)$ the component size is infinite (contributing $z^\infty = 0$ for $z \in [0,1)$), and with probability $\zeta(\lambda)$ the size follows $\mathsf{Borel}(\lambda\zeta(\lambda))$. Therefore,
\begin{equation*}
    \psi_\lambda(z) = \zeta(\lambda) \cdot \BE\left[z^{B(\lambda\zeta(\lambda))}\right].
\end{equation*}
We claim that $\lambda\zeta(\lambda) < 1$, so that $B(\lambda\zeta(\lambda))$ has the standard (subcritical) $\mathsf{Borel}(\lambda\zeta(\lambda))$ distribution. To see this, consider the convex function $f(x) = \mathrm{e}^{-\lambda(1-x)}$ on $[0,1]$. The fixed-point relation~\eqref{eq: extinct_prob} states that $\zeta(\lambda) = f(\zeta(\lambda))$. In the supercritical regime, $f$ has two fixed points in $[0,1]$: the smaller one is $\zeta(\lambda)$ and the larger one is $1$. Since $f$ is convex and $\zeta(\lambda) < 1$, the graph of $f$ must cross the diagonal from above at $x = \zeta(\lambda)$, which means $f'(\zeta(\lambda)) < 1$. Computing $f'(\zeta(\lambda)) = \lambda \mathrm{e}^{-\lambda(1-\zeta(\lambda))} = \lambda\zeta(\lambda)$, we conclude that $\lambda\zeta(\lambda) < 1$.

Now, by the subcritical case, $g(z) \coloneq \BE\left[z^{B(\lambda\zeta(\lambda))}\right]$ satisfies
\begin{equation*}
    g(z) = z \, \mathrm{e}^{\lambda\zeta(\lambda)(g(z)-1)}.
\end{equation*}
Substituting $g(z) = \psi_\lambda(z)/\zeta(\lambda)$ gives
\begin{equation*}
    \frac{\psi_\lambda(z)}{\zeta(\lambda)} = z \, \mathrm{e}^{\lambda\zeta(\lambda)\left(\frac{\psi_\lambda(z)}{\zeta(\lambda)}-1\right)} = z \, \mathrm{e}^{\lambda(\psi_\lambda(z)-\zeta(\lambda))}.
\end{equation*}
Multiplying both sides by $\zeta(\lambda)$ and using $\zeta(\lambda) = \mathrm{e}^{-\lambda(1-\zeta(\lambda))}$ from~\eqref{eq: extinct_prob}, we obtain
\begin{equation*}
    \psi_\lambda(z) = z \, \zeta(\lambda) \, \mathrm{e}^{\lambda(\psi_\lambda(z)-\zeta(\lambda))} = z \, \mathrm{e}^{\lambda(\psi_\lambda(z)-1)}.
\end{equation*}
Finally, the limit $\psi_\lambda(1-) = \zeta(\lambda)$ is consistent with~\eqref{eq: mixture_borel_PGF}, since $\zeta(\lambda) = 1 \cdot \mathrm{e}^{\lambda(\zeta(\lambda)-1)}$ is precisely the fixed-point relation~\eqref{eq: extinct_prob}.\qed
%-----------------------------------------------------------

\subsection{Proof of Theorem~\ref{thm:scaling}}
To prove the theorem, we need the following perturbative result:
\begin{lemma}
\label{lem:perturbative}
Let $z=1-\beta$, and recall that the PGF of the mixture Borel distribution satisfies~\eqref{eq: mixture_borel_PGF}. Then as $\beta\to 0$, for every $x\in\mathbb{R}$:
\begin{equation*}
\psi_{1+x\sqrt\beta}(1-\beta)=1-c(x)\sqrt\beta+o(\sqrt\beta)\,,
\end{equation*}
where $c(x)\coloneq x+\sqrt{x^2+2}$ is the crossover function.
\end{lemma}
\begin{proof}
Fix $x\in \BR$ and set $\lambda_\beta \coloneq 1+x\sqrt\beta$, $z \coloneq 1-\beta$. Let $\psi_\beta \coloneq \psi_{\lambda_\beta}(z)\in(0,1]$, and $\ve_\beta \coloneq 1-\psi_\beta\in[0,1)$. The fixed point $\psi=z e^{\lambda(\psi-1)}$ becomes
\begin{equation*}
1-\ve_\beta=(1-\beta)\exp(-\lambda_\beta\ve_\beta)\,.
\end{equation*}
Taking logs and using $\log(1-y)=-(y+y^2/2)+O(y^3)$ yields
\[
(1-\lambda_\beta)\ve_\beta+\frac{\ve_\beta^2}{2}
=\beta+O(\beta^2+\ve_\beta^3).
\]
Since $\lambda_\beta=1+O(\sqrt\beta)$, the left-hand side is larger than $\ve_\beta^2/4$ for $\beta$ small, hence
$\varepsilon_\beta=O(\sqrt\beta)$. Write $\varepsilon_\beta=c_\beta\sqrt\beta$ and substitute $1-\lambda_\beta=-x\sqrt\beta$:
\[
-xc_\beta+\frac{c_\beta^2}{2}=1+O(\sqrt\beta).
\]
Thus $c_\beta\to c$, where $c$ solves $c^2-2xc-2=0$. Since $\varepsilon_\beta\ge0$, we take the positive root $c(x)=x+\sqrt{x^2+2}$, which justifies the definition of the crossover function.
Therefore $\varepsilon_\beta=c(x)\sqrt\beta+o(\sqrt\beta)$, i.e.
\[
\psi_{1+x\sqrt\beta}(1-\beta)=1-c(x)\sqrt\beta+o(\sqrt\beta),
\]
as claimed.
\end{proof}
\paragraph{Proof of Theorem~\ref{thm:scaling}.} Since $\tau(\alpha,\beta)\to\tau_0\in(0,1)$ as $\beta\to 0$, the parameter $\alpha$ is implicitly a function of $\beta$ along this limit. We may then write~\eqref{eq: asymp_pi_bar} as
\begin{equation*}
    \bar{\pi}_\infty^{\mathrm{global}}(\lambda,\beta)=\frac{A_\beta}{\beta\big(a_\beta+b_\beta\varphi_\lambda(1-\beta)\big)}\,,
\end{equation*}
where
\begin{equation*}
\begin{gathered}
    A_\beta \coloneq \alpha(1-\delta),\quad a_\beta \coloneq (1+\alpha)(1-\delta),\quad b_\beta  \coloneq \delta(\beta-\alpha(1-\beta))\,, \\
    \varphi_\lambda(z)\coloneq \BE\left[z^{B_{\zeta(\lambda)}(\lambda)-1}\right] = \psi_\lambda(z)/z\,.
\end{gathered}
\end{equation*}
Hence
\begin{equation}
\label{eq: pi_diff}
\bar{\pi}_\infty^{\mathrm{global}}(\lambda,\beta)-\bar{\pi}_\infty^{\mathrm{global}}(1,\beta)=-\frac{A_\beta}{\beta}\cdot
\frac{b_\beta\big(\varphi_\lambda(1-\beta)-\varphi_1(1-\beta)\big)}
{\big(a_\beta+b_\beta\varphi_\lambda(1-\beta)\big)\big(a_\beta+b_\beta\varphi_1(1-\beta)\big)}\,.
\end{equation}
From $\tau(\alpha, \beta)=\frac{\alpha}{\beta(1+\alpha)}\to\tau_0\in(0,1)$, we have $\alpha=\tau_0\beta+o(\beta)$. Therefore,
$a_\beta=(1+\alpha)(1-\delta)=(1-\delta)+o(1)$,
$b_\beta=\delta(\beta-\alpha(1-\beta))=\delta (1-\tau_0)\beta +o(\beta)$, and $\frac{A_\beta}{\beta}=\frac{\alpha(1-\delta)}{\beta}=(1-\delta)\tau_0+o(1)$.
Since $\varphi_\lambda(1-\beta) = O(1)$, it follows that $a_\beta+b_\beta\varphi_\lambda(1-\beta)=(1-\delta)+o(1)$, so the product of the two denominator factors in~\eqref{eq: pi_diff} is
$(1-\delta)^2+o(1)$. Plugging these into~\eqref{eq: pi_diff} yields the reduction:
\begin{equation}
\label{eq:diff_appx}
\begin{gathered}
    \bar{\pi}_\infty^{\mathrm{global}}(\lambda,\beta)-\bar{\pi}_\infty^{\mathrm{global}}(1,\beta) = -\frac{\tau_0 (1-\tau_0) \delta}{1-\delta} \, \beta\big(\varphi_\lambda(1-\beta)-\varphi_1(1-\beta)\big) \\
    + o\left( \beta |\big(\varphi_\lambda(1-\beta)-\varphi_1(1-\beta)\big) | \right)\,.
\end{gathered}
\end{equation}
Since $\varphi_\lambda(1-\beta)  = \psi_\lambda(1-\beta)/(1-\beta)$, and $\beta\to 0$, Lemma~\ref{lem:perturbative} gives $\varphi_{1+x\sqrt \beta}(1-\beta) = 1-c(x)\sqrt \beta + o(\sqrt \beta)$, and therefore
\begin{equation*}
\varphi_{1+x\sqrt \beta}(1-\beta) - \varphi_1(1-\beta) = -\big(c(x)-c(0)\big)\sqrt \beta +o(\sqrt \beta) = -\big(c(x)-\sqrt 2\big)\sqrt \beta +o(\sqrt \beta)\,.
\end{equation*}
Inserting this into~\eqref{eq:diff_appx} yields 
\begin{equation*}
\bar{\pi}_\infty^{\mathrm{global}}(\lambda,\beta)-\bar{\pi}_\infty^{\mathrm{global}}(1,\beta) = \frac{\tau_0 (1-\tau_0) \delta}{1-\delta} \, \beta^{3/2} \big(c(x) - \sqrt 2\big) + o(\beta^{3/2})\,.
\end{equation*}
Dividing by $\beta^{3/2}$ and letting $\beta\downarrow0$ gives the claimed limit.\qed

%------------------------------------------------------------------------
\section{Proof of Results in Section~\ref{sec: social_surplus}}
\subsection{Proof of Proposition~\ref{prop: equil_soc_surplus_discont}}
Suppose that initially at $\lambda=\lambda_0$, the common belief lies in $(\pi_{k,n},\pi_{k+1,n}]$, so that $k+1$ agents explore in equilibrium. By~\eqref{eq: asymmetric_equil_cond}, the belief cutoffs $\pi_{k,n}$ are increasing in $\lambda$, so there exists $\lambda_{k,n}(\pi)>\lambda_0$ at which $\pi=\pi_{k,n}$ and the equilibrium transitions to $k$ explorers. Part~\ref{item: Deltau_1} of Lemma~\ref{lem: Delta_u} gives the change in social surplus at $\lambda=\lambda_{k,n}(\pi)$ as $u_{k,n}(\pi)-u_{k+1,n}(\pi)=-\Delta u_k(\pi)$. Substituting $\pi=\pi_{k,n}$ into~\eqref{eq: Deltau_1} yields, at $p=\lambda_{k,n}(\pi)/n$,
\begin{equation*}
    u_{k,n}(\pi_{k,n}) - u_{k+1,n}(\pi_{k,n})
    =-C_2(\pi_{k,n})p(1-p\beta)^{k-1}
    \left[(n-1)(1-p\beta)-k\beta(1-p)\right]\,.
\end{equation*}
The term in brackets is at least $(n-1)(1-\beta)>0$ for $k\leq n-1$. Hence the expression is strictly negative, and the equilibrium social surplus just above $\lambda_{k,n}(\pi)$ is strictly lower than just below it, confirming the discontinuous decline.
%-----------------------------------------
\subsection{Proof of Proposition~\ref{prop: asympt_equil_soc_surplus}}
For $\pi\leq\underline{\pi}$, no agent explores and the social surplus is zero, as in~\eqref{eq:ubar_inf_1}. In the intermediate region $\pi\in(\underline{\pi},\bar{\pi}_\infty^{\mathrm{local}})$, a fraction $k_n/n\to\kappa(\pi)$ of agents explore, as characterized in Proposition~\ref{prop: equil_num_explr_agents}. Since all terms with the $[\cdot]^+$ operator vanish in~\eqref{eq: social_welfare_func}, the average equilibrium social surplus on this region is
\begin{equation}
\label{eq: intermediate_social_surplus}
    \begin{gathered}
       \frac{u_{k_n,n}(\pi)}{n} = (1-\delta)\,\frac{k_n}{n}\left(\pi\beta-\alpha(1-\pi\beta)\right)+\delta\pi(\beta-\alpha(1-\beta))\\
       -\delta\,\frac{k_n}{n}\,\pi(\beta-\alpha(1-\beta))(1-\beta)\BE_{k_n-1}^{(n)}\!\left[(1-\beta)^M\right]-\delta\,\frac{n-k_n}{n}\,\pi(\beta-\alpha(1-\beta))\,\BE_{k_n}^{(n)}\!\left[(1-\beta)^M\right].
    \end{gathered}
\end{equation}
Since $M\sim\mathsf{Bin}(k_n,\lambda/n)$ converges in distribution to $\mathsf{Poisson}(\lambda\kappa(\pi))$ as $n\to\infty$,
\begin{equation*}
    \lim_{n\to\infty}\BE_{k_n-1}^{(n)}\left[(1-\beta)^M\right]=\lim_{n\to\infty}\BE_{k_n}^{(n)}\left[(1-\beta)^M\right]=\ee^{-\lambda\beta\kappa(\pi)}\,.
\end{equation*}
Taking limits in~\eqref{eq: intermediate_social_surplus} and substituting $\kappa(\pi)$ from~\eqref{eq: kappa_expression} yield~\eqref{eq:ubar_inf_2}.

By~\eqref{eq: local_exploration_threshold_limit}, $\bar{\pi}_\infty^{\mathrm{local}}$ is increasing in $\lambda$ with $\bar{\pi}_\infty^{\mathrm{local}}\uparrow\tau$ as $\lambda\to\infty$. Hence for every $\pi\in(\underline{\pi},\tau)$, there exists a threshold $\lambda(\pi)$ such that $\bar{\pi}_\infty^{\mathrm{local}} \geq \pi$ for all $\lambda\geq\lambda(\pi)$, so $\bar{u}_\infty$ follows~\eqref{eq:ubar_inf_2} and is constant in $\lambda$, proving the first claim in Part~\ref{enum:ubar_1}.

Finally, for $\pi\geq\bar{\pi}_\infty^{\mathrm{local}}$, all agents explore ($\kappa(\pi)=1$) and $M$ converges in distribution to $\mathsf{Poisson}(\lambda)$. Using this limiting result in~\eqref{eq: social_welfare_func} yields~\eqref{eq:ubar_inf_3}, proving the formula claims in Parts~\ref{enum:ubar_1} and~\ref{enum:ubar_2}.

For monotonicity, let $a_\pi\coloneq\pi(\beta-\alpha(1-\beta))$. After combining the last two terms of~\eqref{eq:ubar_inf_3}, their integrand $a_\pi(1-(1-\beta)^{m+1})+[a_\pi(1-\beta)^{m+1}-\alpha(1-\pi)]^+$ is nondecreasing in $m$: it is constant while the positive-part term is active and increasing thereafter. Since $\mathsf{Poisson}(\lambda)$ is increasing in $\lambda$ by first-order stochastic dominance,~\eqref{eq:ubar_inf_3} is weakly increasing in $\lambda$. At the switching value $\lambda(\pi)$, the cutoff identity $C_1(\pi)=C_2(\pi)e^{-\lambda(\pi)\beta}$ makes~\eqref{eq:ubar_inf_2} and~\eqref{eq:ubar_inf_3} coincide; since~\eqref{eq:ubar_inf_2} is constant in $\lambda$, the result follows.\qed

%-----------------------------------------
\subsection{Proof of Theorem~\ref{thm:ubar_mixed_infty}}
Recall that the number of explorers is $k^{\mathrm{mixed}}_n \sim \mathsf{Bin}(n,\mu^e_n)$, and set $\xi_n\coloneq k^{\mathrm{mixed}}_n/n$. On the event $k_n^{\mathrm{mixed}}\geq1$, the finite-$n$ average social surplus is, similarly to~\eqref{eq: intermediate_social_surplus},
\begin{equation*}
    \begin{gathered}
       \frac{u_{k^{\mathrm{mixed}}_n,n}(\pi)}{n} = (1-\delta)\,\frac{k^{\mathrm{mixed}}_n}{n}\left(\pi\beta-\alpha(1-\pi\beta)\right)+\delta\pi(\beta-\alpha(1-\beta))\\
       -\delta\,\frac{k^{\mathrm{mixed}}_n}{n}\,\pi(\beta-\alpha(1-\beta))(1-\beta)\BE_{k^{\mathrm{mixed}}_n-1}^{(n)}\!\left[(1-\beta)^M\right]\\
       -\delta\,\frac{n-k^{\mathrm{mixed}}_n}{n}\,\pi(\beta-\alpha(1-\beta))\,\BE_{k^{\mathrm{mixed}}_n}^{(n)}\!\left[(1-\beta)^M\right]\,.
    \end{gathered}
\end{equation*}
Since the two contact counts follow the Binomial distributions $\mathsf{Bin}(k^{\mathrm{mixed}}_n-1,\lambda/n)$ and $\mathsf{Bin}(k^{\mathrm{mixed}}_n,\lambda/n)$, respectively, their PGFs give
\begin{equation*}
    \begin{gathered}
       \frac{u_{k^{\mathrm{mixed}}_n,n}(\pi)}{n} = (1-\delta)\,\xi_n\left(\pi\beta-\alpha(1-\pi\beta)\right)+\delta\pi(\beta-\alpha(1-\beta))\\
       -\delta\,\xi_n\,\pi(\beta-\alpha(1-\beta))(1-\beta)(1-\lambda\beta/n)^{n\xi_n-1}
       -\delta\,(1-\xi_n)\,\pi(\beta-\alpha(1-\beta))\,(1-\lambda\beta/n)^{n\xi_n}.
    \end{gathered}
\end{equation*}
The right-hand side is also zero at $\xi_n=0$; denote it by $f_n(\xi_n)$. Let $f(s)$ be the function obtained by replacing both exponential terms in $f_n(s)$ with $\ee^{-\lambda\beta s}$. For sufficiently large $n$,
\[
\sup_{s\in[0,1]}\left|ns\log(1-\lambda\beta/n)+\lambda\beta s\right|=O(1/n).
\]
It follows that $(1-\lambda\beta/n)^{ns}\to\ee^{-\lambda\beta s}$ uniformly on $[0,1]$. Since
\[
(1-\lambda\beta/n)^{ns-1}=(1-\lambda\beta/n)^{-1}(1-\lambda\beta/n)^{ns},
\]
we also have $f_n\to f$ uniformly on $[0,1]$, and $f$ is continuous.

Realize the sequence $k_n^{\mathrm{mixed}}$ on a common probability space. For every $\varepsilon>0$, Hoeffding's inequality gives
\begin{equation*}
    \BP\left(\left|\xi_n-\mu_n^e\right|>\varepsilon\right)
    \leq 2\ee^{-2n\varepsilon^2}.
\end{equation*}
The upper bound is summable in $n$. By the Borel--Cantelli lemma \citep[Theorem~4.18]{kallenberg2021foundations}, for each $m\in\BN$ the event $|\xi_n-\mu_n^e|>1/m$ occurs only finitely often, almost surely. Taking the intersection of these probability-one events over $m$ yields $\xi_n-\mu_n^e\to0$ almost surely. Since $\mu_n^e\to\kappa(\pi)$ by Proposition~\ref{prop:equil_mix_prob}, we have $\xi_n\to\kappa(\pi)$ almost surely. Consequently,
\begin{equation*}
    \left|f_n(\xi_n)-f(\kappa(\pi))\right|
    \leq \sup_{s\in[0,1]}\left|f_n(s)-f(s)\right|
    +\left|f(\xi_n)-f(\kappa(\pi))\right|\longrightarrow0
\end{equation*}
almost surely. Finally, substituting $\kappa(\pi)$ from~\eqref{eq: kappa_expression} gives $f(\kappa(\pi))=\bar{u}_\infty$.\qed

%-------------------------------------------------------------------------------
\subsection{Proof of Lemma~\ref{lem: B_A_inequality}}
We separately show $B_k$ is larger than both of the arguments of the max operator. First, observe that $B_k \geq (1-\beta)^r A_k$ if and only if
\begin{equation}
\label{eq: B_A_r}
\begin{gathered}
0 \leq k \Big[Q_{k-1}(r-1)(1-\beta)^{r}+\sum_{m=r}^{k-1}q_{k-1}(m)(1-\beta)^{m+1} \Big] \\
-(k+1)\Big[Q_k(r-1)(1-\beta)^{r}+\sum_{m=r}^k q_k(m)(1-\beta)^{m+1} \Big] \\
+(n-k)\Big[Q_k(r)(1-\beta)^r+\sum_{m=r+1}^k q_k(m)(1-\beta)^m\Big]\\
-(n-k-1)\Big[Q_{k+1}(r)(1-\beta)^r +\sum_{m=r+1}^{k+1}q_{k+1}(m)(1-\beta)^m \Big]\,.
\end{gathered}
\end{equation}
Below, $\BE_k$ denotes expectation with respect to $\mathsf{Bin}(k,p)$, where $M$ follows the distribution indicated by the subscript.\footnote{The distribution of $M$ thus varies across terms.}
\begin{equation*}
    \begin{gathered}
    \text{RHS } \text{of \eqref{eq: B_A_r}} =k \BE_{k-1}\left[(1-\beta)^{(M+1) \vee r}\right]-(k+1)\BE_k\left[(1-\beta)^{(M+1)\vee r}\right]\\
    +(n-k) \BE_k\left[(1-\beta)^{M \vee r}\right]- (n-k-1)\BE_{k+1}\left[(1-\beta)^{M \vee r}\right]\,.
    \end{gathered}
\end{equation*}
Since each function inside the expectation operators is decreasing in $M$, first-order stochastic dominance---$\mathsf{Bin}(k,p)\succeq\mathsf{Bin}(k-1,p)$ for the first line and $\mathsf{Bin}(k+1,p)\succeq\mathsf{Bin}(k,p)$ for the second---yields the lower bound:
\begin{equation*}
    \begin{aligned}
        \textit{RHS } \text{of \eqref{eq: B_A_r}} &\geq \BE_{k+1}\left[(1-\beta)^{M \vee r}\right]-\BE_k\left[(1-\beta)^{(M+1)\vee r}\right]\\
        &=p \BE_k \left[(1-\beta)^{(M+1) \vee r}\right]+ (1-p)\BE_k \left[(1-\beta)^{M \vee r}\right]-\BE_k\left[(1-\beta)^{(M+1)\vee r}\right]\\
        &=(1-p)\BE_k\left[(1-\beta)^{M \vee r}-(1-\beta)^{(M+1)\vee r}\right] \geq 0\,,
    \end{aligned}
\end{equation*}
where the second line follows by conditioning $\mathsf{Bin}(k+1,p)$ on its last Bernoulli trial. This establishes the first part.

For the second inequality, $B_k\geq(1-\beta)^{r+1}A_k$, an equivalent condition is:
\begin{equation}
\label{eq: B_A_r+1}
\begin{gathered}
0 \leq k \Big[Q_{k-1}(r-1)(1-\beta)^{r+1}+\sum_{m=r}^{k-1}q_{k-1}(m)(1-\beta)^{m+1} \Big] \\
-(k+1)\Big[Q_k(r-1)(1-\beta)^{r+1}+\sum_{m=r}^k q_k(m)(1-\beta)^{m+1} \Big] \\
+(n-k)\Big[Q_k(r)(1-\beta)^{r+1}+\sum_{m=r+1}^k q_k(m)(1-\beta)^m\Big]\\
-(n-k-1)\Big[Q_{k+1}(r)(1-\beta)^{r+1} +\sum_{m=r+1}^{k+1}q_{k+1}(m)(1-\beta)^m \Big]\,.
\end{gathered}
\end{equation}
Consolidating the right-hand side yields: 
\begin{equation*}
    \begin{gathered}
        \textit{RHS } \text{of \eqref{eq: B_A_r+1}} = (1-\beta)\Big(k\BE_{k-1}\left[(1-\beta)^{M \vee r}\right]-(k+1)\BE_k\left[(1-\beta)^{M \vee r}\right] \\
        +(n-k)\BE_k\left[(1-\beta)^{(M-1)\vee r}\right]-(n-k-1)\BE_{k+1}\left[(1-\beta)^{(M-1)\vee r}\right]
        \Big)\,.
    \end{gathered}
\end{equation*}
Using the first-order stochastic dominance once again yields the following lower bound:
\begin{equation*}
    \begin{aligned}
    \textit{RHS } \text{of \eqref{eq: B_A_r+1}} &\geq (1-\beta) \Big(\BE_{k}\left[(1-\beta)^{(M-1)\vee r}\right]-\BE_k\left[(1-\beta)^{M \vee r}\right]\Big)\\
    &=(1-\beta) \BE_k\left[(1-\beta)^{(M-1) \vee r} - (1-\beta)^{M \vee r}\right] \geq 0\,.
    \end{aligned}
\end{equation*}
This establishes the second part and concludes the proof.\qed

%------------------------------------------------
%---------------------------------------------

\subsection{Proof of Theorem~\ref{thm: social_optimum}}
We first justify the lower cutoff for full exploitation. By part~\ref{item: Deltau_1} of Lemma~\ref{lem: Delta_u}, $\Delta u_0\leq 0$ on $\pi\leq\underline{\pi}^*$ and $\Delta u_0>0$ on $\underline{\pi}^*<\pi\leq\tau$. Since $\Delta u_k$ is decreasing in $k$ on $[0,\tau]$ by Lemma~\ref{lem: decreasing_marginals}, we have $\Delta u_k\leq\Delta u_0\leq 0$ on $\pi\leq\underline{\pi}^*$, establishing the optimality of full exploitation in this region. When $\frac{\pi}{1-\pi}\in\left[\frac{\alpha}{\beta-\alpha(1-\beta)},\frac{\alpha}{(\beta-\alpha(1-\beta))(1-\beta)}\right]$, part~\ref{item: Deltau_3} of Lemma~\ref{lem: Delta_u} (with $r=0$) identifies $\Delta u_0(\pi)$. If $A_0\geq 0$, then $B_0\geq 0$ by Lemma~\ref{lem: B_A_inequality}, and hence
\begin{equation*}
\begin{gathered}
    \Delta u_0(\pi)= (1-\delta)(\pi \beta-\alpha(1-\pi \beta)) + \delta \Big(\pi (\beta-\alpha(1-\beta))\,B_0-\alpha(1-\pi)A_0\Big) \\
    \geq \delta \pi (\beta-\alpha(1-\beta)) \left(B_0 -\frac{\alpha(1-\pi)}{\pi(\beta-\alpha(1-\beta))}\, A_0\right) \geq \delta \pi (\beta-\alpha(1-\beta)) (B_0 -A_0) \geq 0\,.
\end{gathered}
\end{equation*}
where the second to-last inequality holds because $\pi \geq \tau$.

Alternatively, if $A_0<0$, then by Lemma~\ref{lem: B_A_inequality},
\begin{equation*}
\begin{gathered}
    \Delta u_0(\pi) \geq \delta \big(\pi (\beta-\alpha(1-\beta))B_0-\alpha(1-\pi)A_0\big) \\
    =\delta \pi (\beta-\alpha(1-\beta)) \left(B_0 -\frac{\alpha(1-\pi)}{\pi(\beta-\alpha(1-\beta))}\, A_0\right)\\
    \geq \delta \pi (\beta-\alpha(1-\beta)) \left(B_0 - (1-\beta) A_0\right) \geq 0\,,
\end{gathered}
\end{equation*}
where the last inequality holds because $\frac{\pi}{1-\pi} \leq \frac{\alpha}{(\beta - \alpha(1-\beta))(1-\beta)}$.

Finally, when $\frac{\pi}{1-\pi}>\frac{\alpha}{(\beta-\alpha(1-\beta))(1-\beta)}$, part~\ref{item: Deltau_2} of Lemma~\ref{lem: Delta_u} implies $\Delta u_0(\pi)>0$. We conclude that $u_1(\pi)>u_0(\pi)$ for all $\pi>\underline{\pi}^*$, so full exploitation is optimal if and only if $\pi\leq\underline{\pi}^*$.

Next, we establish the optimality of full exploration above $\bar{\pi}^*$. On the region $\pi\leq\tau$, part~\ref{item: Deltau_1} of Lemma~\ref{lem: Delta_u} shows that $\Delta u_{n-1}(\pi)<0$ for $\pi<\bar{\pi}^*$ and $\Delta u_{n-1}(\pi)\geq 0$ for $\pi\in[\bar{\pi}^*,\tau]$. By Lemma~\ref{lem: decreasing_marginals}, $\Delta u_k(\pi)\geq\Delta u_{n-1}(\pi)\geq 0$ for all $\pi\in[\bar{\pi}^*,\tau]$. Thus, within $\pi\leq\tau$, full exploration is optimal if and only if $\pi\geq\bar{\pi}^*$.

It remains to show $\Delta u_k(\pi)\geq 0$ for every $\pi>\tau$. For each such $\pi$, there exists $r$ such that $\frac{\alpha}{(\beta-\alpha(1-\beta))(1-\beta)^r}\leq\frac{\pi}{1-\pi}\leq\frac{\alpha}{(\beta-\alpha(1-\beta))(1-\beta)^{r+1}}$. If $k<r$, part~\ref{item: Deltau_2} gives $\Delta u_k>0$ directly. If $k\geq r$ and $A_k\geq 0$, Lemma~\ref{lem: B_A_inequality} implies
\begin{equation*}
\begin{gathered}
    \Delta u_k(\pi) = (1-\delta)(\pi \beta-\alpha(1-\pi \beta)) + \delta \big(\pi (\beta-\alpha(1-\beta))B_k-\alpha(1-\pi)A_k\big)\\
    \geq \delta \pi (\beta-\alpha(1-\beta)) \left(B_k -\frac{\alpha(1-\pi)}{\pi(\beta-\alpha(1-\beta))}\, A_k\right)\\
    \geq \delta \pi (\beta-\alpha(1-\beta)) \left(B_k - (1-\beta)^r A_k\right)\,.
\end{gathered}
\end{equation*}
Similarly, if $A_k<0$, then
\begin{equation*}
\begin{gathered}
    \Delta u_k(\pi) \geq \delta \big(\pi (\beta-\alpha(1-\beta))B_k-\alpha(1-\pi)A_k\big) \\
    =\delta \pi (\beta-\alpha(1-\beta)) \left(B_k -\frac{\alpha(1-\pi)}{\pi(\beta-\alpha(1-\beta))}\, A_k\right)\\
    \geq \delta \pi (\beta-\alpha(1-\beta)) \left(B_k - (1-\beta)^{r+1} A_k\right) \geq 0\,.
\end{gathered}
\end{equation*}
This confirms that $u_k$ is increasing in $k$ on $[\bar{\pi}^*,1]$, completing the proof.\qed

%%%%%%%%%%%%%%%%%%%%%%%%%%%%%%%%%%%%%%%%%%%%%%%%%%%%%%
\setcitestyle{numbers}	 
\bibliographystyle{normalstyle}
\bibliography{ref}

\end{document}